\documentclass[11pt]{article}

\usepackage[a4paper,margin=27mm]{geometry}
\usepackage{amsmath,amssymb,amsthm,mathtools}
\usepackage{booktabs}
\usepackage{graphicx}
\usepackage{microtype}
\usepackage{placeins}
\usepackage[hidelinks]{hyperref}
\usepackage[nameinlink,noabbrev]{cleveref}
\hypersetup{
  pdftitle={Message Order Selects Opposite Collective Opinions on Laplacian-Cospectral Networks},
  pdfauthor={Ruiwu Niu, Xincheng Shu, and Ying Zhao}
}

\newtheorem{theorem}{Theorem}
\newtheorem{proposition}{Proposition}
\newtheorem{corollary}{Corollary}
\theoremstyle{remark}

\newcommand{\one}{\mathbf 1}

\newcommand{\Rook}{\mathrm R}
\newcommand{\Shri}{\mathrm S}
\newcommand{\Id}{\operatorname{Id}}
\newcommand{\Tr}{\operatorname{Tr}}

\newcommand{\supp}{\operatorname{supp}}
\newcommand{\circprod}{\mathbin{\circ}}

\title{Message Order Selects Opposite Collective Opinions\\
on Laplacian-Cospectral Networks}
\author{%
Ruiwu Niu$^{1}$, Xincheng Shu$^{2,3}$, and Ying Zhao$^{4}$\\[0.5em]
\parbox{0.94\textwidth}{\centering\footnotesize
$^{1}$Department of Data Science and Digital Innovation, Hong Kong Shue Yan University, Hong Kong SAR, China\\
$^{2}$Computational Communication Research Center, Beijing Normal University, Zhuhai 519087, China\\
$^{3}$School of Journalism and Communication, Beijing Normal University, Beijing 100875, China\\
$^{4}$Department of Electrical Engineering, City University of Hong Kong, Kowloon, Hong Kong SAR, China\\[0.3em]
Emails: \texttt{rniu@hksyu.edu} (R. Niu); \texttt{yzhao396-c@my.cityu.edu.hk} (Y. Zhao)}
}
\date{}

\begin{document}
\maketitle

\begin{abstract}
Receiving the same messages in different orders can change an agent's
judgement. We construct a network model in which this local sequence effect
determines the final opinion of an entire network. Each agent has a two-level
quantum-like state with an expressed opinion and an auxiliary context
coordinate that can affect the next update. Two valid update rules of equal
strength leave the neutral state unchanged and add no positive or negative
opinion bias. Nevertheless, reversing their order gives the exact opinions
$+0.01$ and $-0.01$ from the same initial state. We place four ordered pairs so
that their positive and negative contributions sum to zero but occupy different
labelled nodes. After this single input, all cases follow the same nonlinear
network rule, which preserves valid states and admits two stable consensuses.
We compare two 16-node networks whose complete Laplacian spectra are identical.
With every local ingredient fixed, they converge to opposite consensuses. A
separate local calculation identifies a specific overlap between the labelled
pulse and Laplacian eigenspace projectors that sets the leading displacement of
the basin boundary. Rigorous interval arithmetic proves the finite registered
example independently. The construction connects message order to collective
selection and shows why Laplacian eigenvalues alone do not determine which
opinion wins.
\end{abstract}

\section{Introduction}

Arguments, evidence and frames arrive sequentially. Their order can therefore
be part of the interaction rather than a cosmetic feature of presentation.
This issue is especially visible in networks of large language model agents,
where recent experiments report networked consensus and fragmentation,
sensitivity to retained dialogue history, model-dependent persuasion under
different framings, and dependence on the order in which debate answers are
presented
\cite{ChuangEtAl2024,CisnerosVelarde2024,CauEtAl2025,
StengelEskinEtAl2025}.
The present paper does not fit an empirical model to those systems. It develops
a theory that isolates one mechanism through which local message order can be
amplified into a collective decision.

Quantum-like models provide a compact language for sequence-sensitive state
updates. A state records both the current readout and contextual information
that is not directly visible in that readout. Two updates need not commute, so
applying update $A$ and then $B$ can differ from applying $B$ and then $A$.
This operational use of a density matrix organizes preparations, updates and
measurements; it does not require a person or an artificial agent to be a
physical quantum system. Question-order models and quantum instruments already
establish this local principle
\cite{WangBusemeyer2013,WangEtAl2014,OzawaKhrennikov2021,FuyamaEtAl2025}.

Order sensitivity is not specific to quantum-like formalisms. Classical rules
for opinion revision can exhibit primacy and recency, and changing only the
schedule of agent updates can alter collective outcomes
\cite{AllahverdyanGalstyan2014,WeimerEtAl2019}. The present use of a rebit is
therefore a compact operational model of the local state and its updates, not a
claim that order effects require quantum theory.

Network versions of quantum-like opinion dynamics are also established.
Prior work includes quantum walks for idea propagation, quantum-theoretic
opinion models, and consensus in density-operator or other noncommutative
spaces \cite{ZhangBusemeyer2021,SansEtAl2026,ShiEtAl2016,
Jafarizadeh2016,SepulchreEtAl2010}. Most directly, Chu assigns density matrices
to network nodes, represents order effects with noncommuting question
operators, couples nodes through Lindblad dynamics, and recovers a
Friedkin--Johnsen model under a product-state reduction
\cite{FriedkinJohnsen1990,Chu2026}. Our novelty claim is therefore not based on
density states, noncommutativity, order effects or Laplacian coupling by
themselves.

Operator order has also been incorporated into social network group decisions,
and quantum-like state evolution has been used to model repeated influence
among several respondents \cite{ShenEtAl2024,BroekaertEtAl2026}. These studies
further narrow the present contribution to the particular construction that
links a finite ordered channel pair to a balanced network pulse and then to a
nonlinear choice between attractors.

Recent studies have also placed two-level or qubit-like cognitive states on
interacting-agent networks, using Hamiltonian coupling, quantum-inspired
polarization mechanisms, quantum-circuit simulations and Bloch-sphere
collective dynamics
\cite{AlodjantsEtAl2024,AlodjantsEtAl2025,MaksymovPogrebna2024,GuoEtAl2025,
SansEtAl2026PLA}. These works establish quantum-like network opinion states and
collective phases as prior art. Their primary mechanisms differ from the
construction here, in which channel order generates a balanced pulse.

The distinct problem addressed here begins after the local order effect is
created. We study how a fixed ordered pair writes a balanced, labelled pulse into a
nonlinear network and how that input selects one of two stable collective
opinions. Classical network models already show that local influence and
topology can create agreement, disagreement and multistability
\cite{DevriendtLambiotte2021,HomsDonesEtAl2021,BizyaevaEtAl2023,
ParkEtAl2022}. It is also known that an eigenvalue list does not determine a
graph. Matched spectra can coexist with different synchronization transients or
convergence rates
\cite{vanDamHaemers2003,RavooriEtAl2011,FishSun2017,NazerianEtAl2024,
MisiakDziubinski2026}. Separately, network structure can change nonlinear
phase-locked states and their basins \cite{SokolovErmentrout2019}. These
precedents make the relevant target narrower. The contribution is the verified
composition of the ingredients: an exact order-sensitive channel pair creates
a balanced labelled pulse, the same nonlinear flow sends that pulse to
opposite attractors on a fixed Laplacian-cospectral pair, and a specific
projector overlap determines the leading local displacement of the basin
boundary.

For one explicit model, the paper makes three contributions.
\begin{enumerate}
\item We construct two finite quantum-like update channels of equal strength.
They contain no additive opinion bias, yet their two orders produce the exact
opinions $+0.01$ and $-0.01$ from the same preparation. Restriction to the
population sector and a parity-twirl control each remove this contrast and
identify the contextual pathway within the declared model.
\item We place four ordered pairs on a network so that the injected opinions
sum to zero, then evolve the full node states with one autonomous nonlinear
master equation. Expressed opinion is an exact coordinate of this flow rather
than an appended approximation. The typed channels act only during the atomic
input step; the subsequent sender feedback is classical and common to all
graphs and order arms.
\item We hold the channels, preparations, labels, input roles and network rule
fixed on the rook and Shrikhande graphs. Although these graphs have the same
complete Laplacian spectrum $\{0^{(1)},4^{(6)},8^{(9)}\}$, they reach opposite
consensus states. A local calculation of the stable manifold identifies a
specific overlap of complete Laplacian eigenspace projectors in the leading
basin response. A separate interval proof, computed with outward rounding,
certifies the winner table at finite amplitude.
\end{enumerate}

The proof has a simple dependency structure. The local channel theorem creates
the balanced pulse. The network theorem shows that this pulse enters a valid
dynamics of the complete state with an exact opinion coordinate. The theorem
on opposite winners then lifts a rigorous finite-amplitude opinion certificate
back to the complete node states. Separately, the projector calculation
explains the topology dependence within the small-amplitude regime for which it
is proved.

\Cref{fig:mechanism} gives the whole argument in three stages. Its central
bridge is that a zero-sum input need not vanish: the signs can still form a
nonzero pattern across labelled nodes, leaving information for the network
topology to act on.

\begin{figure}[t!]
\centering
\IfFileExists{figures/fig1_mechanism.pdf}{%
  \includegraphics[width=\linewidth]{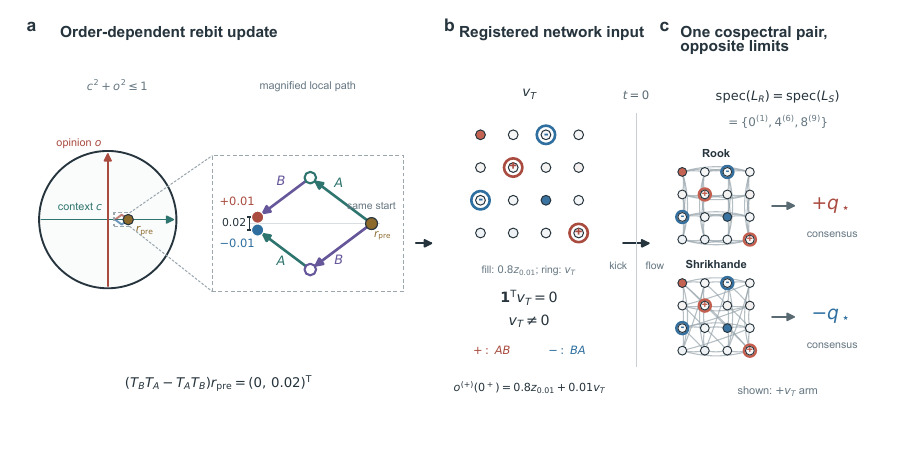}%
}{%
  \fbox{\parbox[c][0.25\textheight][c]{0.94\linewidth}{\centering
  Placeholder for the local-to-network mechanism figure.\\[0.5em]
  Rebit state $\to$ ordered $A/B$ updates $\to$ balanced pulse
  $\to$ common network flow.}}
}
\caption{\textbf{A local order effect becomes a topology-dependent collective
outcome.} Panel a shows the rebit Bloch disk $c^2+o^2\leq1$, where
$c,o\in[-1,1]$ denote auxiliary context and expressed opinion. Both paths start
from the same valid vector $r_{\rm pre}=(25/84,0)^{\mathsf T}\in\mathbb R^2$.
The left dashed box selects the path region, and the dashed frame on the right
encloses an enlargement with equal display scales for $c$ and $o$.
The matrices $T_A,T_B\in\mathbb R^{2\times2}$ are the Bloch-plane actions of
the two channels of equal strength. The upper path applies $A$ first and $B$
second, while the lower path reverses the order. Open circles mark intermediate
states; filled endpoints have opinions $+0.01$ and $-0.01$. Their exact
difference is the displayed commutator action. Panel b shows the complete
16-node input. Node interiors encode the common nonuniform background
$0.8z_{0.01}$, where $z_{0.01}\in[-1,1]^{16}$ is a fixed reference profile.
The outer rings encode the order pulse
$v_T\in\{-1,0,1\}^{16}$: red rings receive $AB$ and blue rings receive $BA$.
For the vector of ones $\mathbf 1\in\mathbb R^{16}$,
$\mathbf 1^{\mathsf T}v_T=0$ although $v_T\ne0$. Thus the pulse has zero total
but retains its node locations. With the superscript $+$ selecting the shown
$+v_T$ assignment, the opinion vector after the kick is
$o^{(+)}(0^+)=0.8z_{0.01}+0.01v_T\in(-1,1)^{16}$. The ordered updates stop at
$t=0$, after which only the common network flow acts. Panel c applies this same
labelled input and the same parameters to the rook and Shrikhande networks.
Their Laplacian matrices $L_R,L_S\in\mathbb R^{16\times16}$ have identical
complete spectra, yet the shown $+v_T$ arm converges to $+q_\star$ on the rook
network and to $-q_\star$ on the Shrikhande network, where
$q_\star\in(0,1)$ is the consensus magnitude. Reversing all four ring
assignments reverses both signs. \Cref{fig:witness} separately displays the
small-amplitude projector mechanism, the finite-amplitude interval certificate
and the parity-twirled control.}
\label{fig:mechanism}
\end{figure}
\FloatBarrier

Panel b connects the local update to the network flow used below. Scalar
balance removes net bias but preserves the locations of the positive and
negative entries. At $t=0$, the protocol switches from the single channel
update to a common flow that preserves state validity and alone governs
interactions along network edges.

\section{Operational rebit and an exact order pulse}
\label{sec:source}

\subsection{State and finite update channels}

The local state needs to retain more than the opinion that is read out at one
instant. We therefore use one visible coordinate for opinion and one auxiliary
coordinate for contextual information that a later message can convert into
opinion. The smallest real quantum-like representation with these two
coordinates is a rebit, namely a two-level density matrix restricted to the
real Bloch plane.

Let $V$ be a finite node set, let $N=|V|$, and use $i,j\in V$ for node
indices. Let $\mathbb M_2=\mathbb C^{2\times2}$ denote the algebra of complex
$2\times2$ matrices. Write $I_2\in\mathbb R^{2\times2}$ for the identity matrix,
$\sigma_x=\left(\begin{smallmatrix}0&1\\1&0\end{smallmatrix}\right)$ and
$\sigma_z=\left(\begin{smallmatrix}1&0\\0&-1\end{smallmatrix}\right)$ for the
two real Pauli matrices, and $\Tr$ for the matrix trace. The allowed rebit
states form the following subset of density matrices:
\[
\mathcal D_{\rm R}=\{\rho\in\mathbb R^{2\times2}:\rho=\rho^\top,
\ \rho\succeq0,\ \Tr\rho=1\}\subset\mathbb M_2.
\]
At node $i\in V$, let $\rho_i\in\mathcal D_{\rm R}$ be parameterized by
$c_i,o_i\in\mathbb R$ and write
\begin{equation}
\begin{aligned}
\rho_i&=\frac12\left(I_2+c_i\sigma_x+o_i\sigma_z\right),
&r_i&=(c_i,o_i)^\top,\\
r_i&\in\mathbb B_2:=\{r\in\mathbb R^2:\|r\|_2\leq1\},
&\|r_i\|_2&\leq1.
\end{aligned}
\label{eq:rebit}
\end{equation}
The inequality is exactly the condition that $\rho_i$ is a valid state. The
Bloch disk in \cref{fig:mechanism}a is simply the set of all allowed pairs
$(c_i,o_i)$. The readout $o_i=\Tr(\sigma_z\rho_i)$ is expressed opinion, while
$c_i=\Tr(\sigma_x\rho_i)$ is the auxiliary context coordinate. The latter is
not reported as opinion, but it is causally active in the ordered updates.

The formal message types $A$ and $B$ act along two different axes of this
disk. A nonselective dephasing channel preserves the component parallel to its
axis and removes the perpendicular component. We use a finite partial version
that retains half of the perpendicular component, so each update is nontrivial
but not a complete projection.
Choose two unit axes $n_A,n_B\in S^1:=\{n\in\mathbb R^2:\|n\|_2=1\}$
in the real Bloch plane,
\[
n_A=\left(\frac35,\frac45\right)^\top,
\qquad
n_B=\left(\frac35,-\frac45\right)^\top,
\]
and, for a message type $\ell\in\{A,B\}$, define the rank-one measurement
projectors
\[
\Pi_\ell^\pm
=\frac12\left[I_2\pm n_\ell\mathbin{\cdot}(\sigma_x,\sigma_z)\right],
\qquad \ell\in\{A,B\}.
\]
For $n=(n_1,n_2)^\top\in\mathbb R^2$, the matrix-valued dot product means
$n\mathbin{\cdot}(\sigma_x,\sigma_z)=n_1\sigma_x+n_2\sigma_z$.
The associated nonselective dephasing channel
$\mathcal D_\ell:\mathbb M_2\to\mathbb M_2$ is
\[
\mathcal D_\ell(\rho)
=\Pi_\ell^+\rho\Pi_\ell^+
+\Pi_\ell^-\rho\Pi_\ell^-.
\]
It preserves $\mathcal D_{\rm R}$. Writing $\Id$ for the identity superoperator
on $\mathbb M_2$, we use
the finite partial-dephasing channel
\begin{equation}
\mathcal E_\ell=\frac12\Id+\frac12\mathcal D_\ell.
\label{eq:partial-dephasing}
\end{equation}
It is completely positive and trace preserving, abbreviated CPTP, so it maps
every valid input state to a valid output state. It is also unital, meaning
that it leaves the neutral state $I_2/2$ unchanged and therefore adds no affine
opinion bias. Both channels have the same finite strength and preserve the
real Bloch plane.

Writing $P_\ell=n_\ell n_\ell^\top\in\mathbb R^{2\times2}$, the Bloch action is
$r\mapsto T_\ell r$ with $T_\ell=(I_2+P_\ell)/2$. Explicitly,
\begin{equation}
T_A=\begin{pmatrix}
17/25&6/25\\[2pt]
6/25&41/50
\end{pmatrix},
\qquad
T_B=\begin{pmatrix}
17/25&-6/25\\[2pt]
-6/25&41/50
\end{pmatrix}.
\label{eq:bloch-maps}
\end{equation}

\begin{theorem}[Exact order pulse mediated by context]
\label{thm:local-source}
The channels $\mathcal E_A$ and $\mathcal E_B$ are finite CPTP maps that
preserve rebit states and satisfy
\begin{equation}
T_BT_A-T_AT_B
=\frac{42}{625}
\begin{pmatrix}0&-1\\[2pt]1&0\end{pmatrix}.
\label{eq:commutator}
\end{equation}
From the common preparation $r_{\mathrm{pre}}\in\mathbb B_2$,
\[
r_{\mathrm{pre}}=\left(\frac{25}{84},0\right)^\top,
\]
the order $AB$, meaning first $A$ and then $B$, and the reverse order give
\begin{equation}
T_BT_A r_{\mathrm{pre}}
=\left(\frac{253}{2100},+\frac1{100}\right)^\top,
\qquad
T_AT_B r_{\mathrm{pre}}
=\left(\frac{253}{2100},-\frac1{100}\right)^\top.
\label{eq:exact-source}
\end{equation}
Both outputs are valid rebit states and have equal contextual endpoints.
\end{theorem}

\begin{proof}
Equation~\eqref{eq:partial-dephasing} is a convex mixture of CPTP maps, so it
is CPTP. Direct multiplication of the rational matrices in
\eqref{eq:bloch-maps} gives \eqref{eq:commutator} and
\eqref{eq:exact-source}. The preparation has norm $25/84<1$, while either
output has squared norm $1289/88200<1$.
\end{proof}

\paragraph{Mechanism controls.}
For the population restriction
$\Delta_Z:\mathbb R^2\to\mathbb R^2$, $\Delta_Z(c,o)=(0,o)$, and the opinion
unit vector $e_o=(0,1)^\top$,
\[
\Delta_ZT_A\Delta_Z=\Delta_ZT_B\Delta_Z,
\qquad
e_o^\top T_Ae_o=e_o^\top T_Be_o=\frac{41}{50}.
\]
Since $\Delta_Zr_{\mathrm{pre}}=0$, both restricted orders give zero opinion.
This restriction is applied before and throughout the update sequence.
Projecting only after the first full update is a different operation because
the first channel can transfer contextual information into opinion.

For the parity superoperator
$\mathcal Z:\mathbb M_2\to\mathbb M_2$,
$\mathcal Z(\rho)=\sigma_z\rho\sigma_z$, define
\[
\overline{\mathcal E}_\ell
=\frac12\left(\mathcal E_\ell+
\mathcal Z\circ\mathcal E_\ell\circ\mathcal Z\right).
\]
The two twirled Bloch maps coincide:
\begin{equation}
\overline T_A=\overline T_B
=\operatorname{diag}\left(\frac{17}{25},\frac{41}{50}\right).
\label{eq:twirl}
\end{equation}
The unchanged preparation has zero opinion, so both twirled orders have zero
opinion and zero order contrast. This establishes that the nonzero contrast
in \eqref{eq:exact-source} is mediated by the context coordinate relative
to the declared preparation and readout basis.

Both paths start from zero expressed opinion, contain exactly one update of
each type and end with the same context coordinate. The controls identify
transfer from the context coordinate to opinion as the source of the sign
reversal for this preparation and readout. This mechanism statement applies to
the declared operational model. Comparison with general classical latent-state
models lies outside its scope.

\subsection{Crossed injection and the atomic protocol boundary}

The local order effect must next be written into a network while keeping the
two message orders equally represented. We use four active nodes. Two receive
$AB$ and two receive $BA$, with the assignments swapped in a second arm. The
resulting signed opinion pulse sums to zero.

For the explicit witness, specialize to the labelled vertex set
$V=\mathbb Z_4^2$ and identify $\mathbb R^V\cong\mathbb R^{16}$ with a
$4\times4$ array whose rows and columns carry the vertex labels. Define three
fixed vectors $z_0,p,v_T\in\mathbb R^{16}$ by
\begin{equation}
z_0=
\begin{pmatrix}
1&0&0&0\\0&0&0&0\\0&0&-1&0\\0&0&0&0
\end{pmatrix},
\quad
p=
\begin{pmatrix}
0&1&0&2\\3&0&4&5\\0&-2&0&-1\\-4&-5&-3&0
\end{pmatrix},
\quad
v_T=
\begin{pmatrix}
0&0&-1&0\\0&1&0&0\\-1&0&0&0\\0&0&0&1
\end{pmatrix},
\label{eq:p-vector}
\end{equation}
and, for a real deformation parameter $\tau\in\mathbb R$, set
\begin{equation}
z_\tau=z_0+\tau p.
\label{eq:center-source}
\end{equation}
Here $z_\tau\in\mathbb R^{16}$ is a zero-sum background opinion pattern, and
$v_T$ marks the four nodes that receive ordered message pairs. Both arrays are
balanced, they are orthogonal, and $z_\tau$ vanishes on the set of active nodes
$\supp(v_T):=\{i\in V:v_{T,i}\neq0\}$.
The coefficients in $p$ define an engineered nonconjugate witness; they are
not estimated from data. Let $\varrho\in[0,1]$ denote the amplitude of the
background state and let $\mu\geq0$ denote the pulse amplitude. For the
registered witness, freeze
\begin{equation}
\varrho=\frac45,
\qquad
\tau=\frac1{100},
\qquad
\mu=\frac1{100}.
\label{eq:registered-source}
\end{equation}
A preparation at an inactive node is admissible whenever
$\varrho\|z_\tau\|_\infty\leq1$. At the registered point,
$\|z_{1/100}\|_\infty=1$, so this product is $4/5<1$.

In arm $\eta\in\{-1,+1\}$, apply $AB$ at active node $i$ when
$\eta v_{T,i}=+1$, and apply $BA$ when $\eta v_{T,i}=-1$. Each arm therefore
contains two blocks of each order, and every active node receives the same
message multiset $\{A,B\}$.

\begin{center}
\fbox{\begin{minipage}{0.91\linewidth}
\textbf{Atomic kick followed by network flow.}
\textbf{1. Prepare.} Both graphs and both arms use the same node states at
$0^-$. Active nodes are in $r_{\mathrm{pre}}$; inactive nodes have Bloch vector
$(0,\varrho z_{\tau,i})^\top$.\par
\textbf{2. Kick.} The four blocks, each containing two messages, are executed
synchronously as one
indivisible map $K_\eta:\mathcal D_{\rm R}^{16}\to\mathcal D_{\rm R}^{16}$ on
the full network state
$\mathbf R=(\rho_1,\ldots,\rho_{16})\in\mathcal D_{\rm R}^{16}$.
No network evolution, readout, adaptation or
additional message occurs during the kick.\par
\textbf{3. Flow.} The symbols $0^-$ and $0^+$ denote the instants immediately
before and after the kick. For continuous time $t\in(0,\infty)$, every graph and arm follows
the same autonomous network rule introduced in \cref{sec:flow}.
\end{minipage}}
\end{center}

By \cref{thm:local-source}, the complete state immediately after the kick is
\begin{equation}
c_i(0^+)=
\begin{cases}
253/2100,&i\in\supp(v_T),\\
0,&i\notin\supp(v_T),
\end{cases}
\qquad
o^{(\eta)}(0^+)=\varrho z_\tau+\eta\mu v_T,
\label{eq:postkick}
\end{equation}
where $o^{(\eta)}=(o_i^{(\eta)})_{i\in V}\in[-1,1]^{16}$ is the opinion vector
after the kick in arm $\eta$.
The opinion perturbation $\eta\mu v_T$ is balanced. The nonzero contextual
component is identical in both arms and therefore does not encode the arm
sign. Channel order is the only operation that depends on the arm. The message types,
their strengths, the node labels, the preparation profile and the multiset at
each active node remain fixed.

\section{A valid nonlinear network flow for the full state}
\label{sec:flow}

\subsection{Replacement master equation}

After the atomic kick, the model needs a network rule that accepts both output
coordinates, keeps every node state valid and supports more than one collective
outcome. The rule below has a simple interpretation in terms of senders and
receivers.
A sender with opinion $o_j$ emits a positive or negative stance $s$. The
receiver is pulled toward the corresponding stance state $\Omega_s$, while a
separate relaxation term pulls it toward the neutral state $\Omega_0$. The
parameter $\beta$ controls how strongly a sender's existing opinion is
amplified when generating its stance.

For a stance sign $s\in\{-1,+1\}$, define the corresponding target state
$\Omega_s\in\mathcal D_{\rm R}$ and the neutral target $\Omega_0\in\mathcal D_{\rm R}$ by
\[
\Omega_s=\frac12(I_2+s\sigma_z),
\qquad
\Omega_0=\frac {I_2}2.
\]
For a parameter $\beta\in[1,\infty)$ that controls stance sharpening, define the conditional
probability $q_\beta:\{-1,+1\}\times[-1,1]\to[0,1]$ and its mean stance
$g_\beta:[-1,1]\to[-1,1]$ by
\begin{equation}
q_\beta(s\mid o)
=\frac{(1+so)^\beta}{(1+o)^\beta+(1-o)^\beta},
\qquad
g_\beta(o)=\sum_{s=\pm1}s q_\beta(s\mid o).
\label{eq:sender-policy}
\end{equation}
Equivalently,
$g_\beta(o)=\tanh[\beta\operatorname{artanh}(o)]$, with the endpoint values
defined continuously. When $\beta=1$, the mean transmitted stance equals the
sender's opinion. Values $\beta>1$ sharpen an existing positive or negative
opinion without choosing a preferred sign.

For $s\in\{-1,0,+1\}$ and a matrix $X\in\mathbb M_2$, let the linear map
$\mathcal L_s:\mathbb M_2\to\mathbb M_2$ be
\[
\mathcal L_s(X)=\Tr(X)\Omega_s-X.
\]
Each $\mathcal L_s$ is a replacement generator of
Gorini--Kossakowski--Lindblad--Sudarshan (GKLS) form
\cite{GoriniEtAl1976,Lindblad1976}.
For a nonnegative symmetric weight matrix
$W=(w_{ij})\in\mathbb R_{\geq0}^{N\times N}$, let $w_{ij}$ be the influence
weight from node $j$ to node $i$ and write the weighted degree as
$d_i=\sum_{j\in V}w_{ij}\in[0,\infty)$. Let
$G=(V,W)$ denote the weighted undirected network encoded by $W$. Let
$\alpha,\kappa\in[0,\infty)$ be, respectively, the rate of influence from
neighbors and the rate of relaxation to the neutral state. For time
$t\in[0,\infty)$, define
\begin{equation}
\boxed{
\dot\rho_i
=\alpha\sum_{j\in V}w_{ij}\sum_{s=\pm1}
q_\beta(s\mid o_j)\mathcal L_s(\rho_i)
+\kappa\mathcal L_0(\rho_i).
}
\label{eq:density-flow}
\end{equation}
The influence rate $\alpha$ scales neighbor input, and $\kappa$ scales return
to the neutral state. The global equation is nonlinear because the sending
probabilities depend on the current opinions. At any fixed network state,
however, a receiver is updated by a nonnegative mixture of valid replacement
generators.

\begin{theorem}[Preservation of density states and an exact opinion coordinate]
\label{thm:network-flow}
Suppose $w_{ij}\geq0$, $W=W^\top$, $\alpha,\kappa\geq0$ and $\beta\geq1$.
For every collection $(\rho_1,\ldots,\rho_N)$ of rebit density matrices in
the Cartesian state space, \eqref{eq:density-flow} has a unique global forward
solution. It preserves positivity, unit trace, Hermiticity and the real Bloch
plane at every node. The Bloch coordinates satisfy exactly
\begin{align}
\dot c_i&=-(\alpha d_i+\kappa)c_i,
\label{eq:context-flow}\\
\dot o_i&=\alpha\left[-d_i o_i+
\sum_jw_{ij}g_\beta(o_j)\right]-\kappa o_i.
\label{eq:opinion-flow}
\end{align}
If $\Phi_G^t:\mathcal D_{\rm R}^N\to\mathcal D_{\rm R}^N$ is the flow on
density states, $\phi_G^t:[-1,1]^N\to[-1,1]^N$ is the flow of
\eqref{eq:opinion-flow}, and
$\pi_o:\mathcal D_{\rm R}^N\to[-1,1]^N$ is the opinion projection
\[
\pi_o(\rho_1,\ldots,\rho_N)
=\bigl(\Tr(\sigma_z\rho_i)\bigr)_{i\in V},
\]
then these two flows, both generated by the weight matrix $W$ of $G$, are
linked by the exact opinion semiconjugacy
\begin{equation}
\pi_o\circ\Phi_G^t=\phi_G^t\circ\pi_o,
\qquad
c_i(t)=e^{-(\alpha d_i+\kappa)t}c_i(0).
\label{eq:semiconjugacy}
\end{equation}
\end{theorem}

\begin{proof}
Put $\gamma_i=\alpha d_i+\kappa\in[0,\infty)$, the total replacement rate at
node $i$. When $\gamma_i>0$, write
$\boldsymbol\rho=(\rho_j)_{j\in V}\in\mathcal D_{\rm R}^N$ for the full state
and define the target map
$\Theta_i:\mathcal D_{\rm R}^N\to\mathcal D_{\rm R}$ by
\begin{equation}
\Theta_i(\boldsymbol\rho)=
\frac{
\alpha\sum_jw_{ij}\sum_{s=\pm1} q_\beta(s\mid o_j)\Omega_s
+\kappa\Omega_0}{\gamma_i}.
\label{eq:target-state}
\end{equation}
The numerator is a nonnegative mixture whose coefficients sum to
$\gamma_i$, so $\Theta_i(\boldsymbol\rho)$ is a density matrix. Equation
\eqref{eq:density-flow} becomes
$\dot\rho_i=\gamma_i[\Theta_i(\boldsymbol\rho)-\rho_i]$, and variation of constants gives
\begin{equation}
\rho_i(t)=e^{-\gamma_it}\rho_i(0)
+\int_0^t\gamma_i e^{-\gamma_i(t-s)}\Theta_i(\boldsymbol\rho(s))\,ds.
\label{eq:convex-solution}
\end{equation}
The weights in \eqref{eq:convex-solution} form a convex combination. If
$\gamma_i=0$, the node is stationary. The vector field also preserves trace,
Hermiticity and the real Bloch plane. For $\beta\geq1$, $q_\beta$ is
Lipschitz on the closed opinion interval. Local existence, the convex
representation and a first-exit argument give a unique global forward
solution in the Cartesian product of the node state spaces.

Taking the $\sigma_x$ and $\sigma_z$ traces of
\eqref{eq:density-flow} yields \eqref{eq:context-flow} and
\eqref{eq:opinion-flow}. These equations prove
\eqref{eq:semiconjugacy}.
\end{proof}

For a 6-regular graph $G$ with adjacency matrix
$A_G\in\{0,1\}^{N\times N}$ and opinion vector
$o=(o_i)_{i\in V}\in[-1,1]^N$, the registered values
$\alpha=\kappa=1$ and $\beta=6/5$ give
\begin{equation}
\boxed{
\dot o=-7o+A_Gg_{6/5}(o),
}
\label{eq:registered-flow}
\end{equation}
where $g_{6/5}$ acts componentwise. The theorem closes the interface between
the valid density state after the kick and this exact opinion coordinate. The
$A/B$ channels act only in the atomic kick; afterward the common classical
sender policy drives the network and context decays without feeding back into
opinion.

\section{Why eigenvalues are not enough: spectral projector geometry}
\label{sec:projectors}

The Laplacian separates collective motion into network modes. An eigenvalue
determines how quickly its mode grows or decays. A spectral projector onto the
corresponding full eigenspace determines the spatial part of a labelled input
that occupies that mode. The second object matters whenever the node labels
and injection locations are fixed. It is also basis independent when an
eigenvalue is repeated.

\subsection{Modal rates and labelled spatial alignment}

Let $G$ be a connected $d$-regular graph on $N$ nodes, where
$d\in\mathbb N$, $A\in\mathbb R^{N\times N}$ is its adjacency matrix, and
$I_N\in\mathbb R^{N\times N}$ is the identity on node space. Write
$\one=\one_N=(1,\ldots,1)^\top\in\mathbb R^N$ for the vector of ones, and let
$L=dI_N-A\in\mathbb R^{N\times N}$ be the graph Laplacian. For a time horizon
$T>0$, let $q:[0,T]\to[-1,1]$ define the homogeneous opinion trajectory
$o(t)=q(t)\one$. Linearization of
\eqref{eq:opinion-flow} gives
\begin{equation}
J_L(t)=
\left[\alpha d\{g_\beta'(q(t))-1\}-\kappa\right]I_N
-\alpha g_\beta'(q(t))L.
\label{eq:homogeneous-jacobian}
\end{equation}
If $\lambda\in\operatorname{spec}(L)$ is a distinct Laplacian eigenvalue and
$P_\lambda\in\mathbb R^{N\times N}$ is the orthogonal projector onto its full
eigenspace, then the propagator from time $t_0$ to $T$, with
$0\leq t_0\leq T$, is
\begin{align}
\Phi_L(T,t_0)&=\sum_\lambda H_\lambda(T,t_0)P_\lambda,
\label{eq:projector-response}\\
H_\lambda(T,t_0)&=
\exp\!\left(\int_{t_0}^T
\left[\alpha d\{g_\beta'(q(t))-1\}-\kappa
-\alpha g_\beta'(q(t))\lambda\right]dt\right).
\end{align}
For the full rebit variational equation, the corresponding expression is
\begin{equation}
\Phi_L^{\mathrm{full}}(T,t_0)
=\sum_\lambda P_\lambda\otimes
\begin{pmatrix}
e^{-(\alpha d+\kappa)(T-t_0)}&0\\
0&H_\lambda(T,t_0)
\end{pmatrix}.
\label{eq:full-projector-response}
\end{equation}
The positive scalar $H_\lambda(T,t_0)$ is the modal gain. These gains depend
only on the eigenvalues. The factors
$P_\lambda$ record how a fixed labelled pulse is distributed over the spatial
modes. Here $\lambda_2$ denotes the smallest positive Laplacian eigenvalue of a
connected graph. Consequently, equality of $\lambda_2$, or even equality of the complete
eigenvalue list including multiplicities, does not fix the response to a
labelled perturbation.

Equations \eqref{eq:homogeneous-jacobian} to
\eqref{eq:full-projector-response} apply only along a homogeneous reference
trajectory. They are not used to prove the moderate-amplitude winner theorem
in \cref{sec:witness}.

\subsection{The local basin boundary}

For the perturbation family considered below, the local stable manifold of the
unstable neutral state acts as the decision surface between the two stable
consensus outcomes. The coordinate $y$ below measures motion along the
consensus direction, while $z$ is a disagreement pattern with zero mean.

For the scalar opinion system on a $d$-regular graph, define the consensus
growth rate $\lambda_u\in\mathbb R$ and the unit consensus direction
$u\in\mathbb R^N$ by
\[
\lambda_u=\alpha d(\beta-1)-\kappa,
\qquad
u=\frac{\one}{\sqrt N}.
\]
Assume
\begin{equation}
0<\lambda_u<\alpha\beta\lambda_2.
\label{eq:index-one}
\end{equation}
The origin is then an index-one saddle. Its consensus rate is $\lambda_u$,
while a
transverse Laplacian mode $\lambda>0$ has rate
$\sigma_\lambda=\lambda_u-\alpha\beta\lambda<0$. Writing $o=yu+z$ with
$y\in\mathbb R$ and $z\in u^\perp:=\{x\in\mathbb R^N:u^\top x=0\}$, there is
a neighborhood $U_G\subset u^\perp$ and an odd graph function
$h_G:U_G\to\mathbb R$ whose cubic term is $h_{3,G}:u^\perp\to\mathbb R$:
\cite{HirschPughShub1977}
\[
y=h_G(z)=h_{3,G}(z)+O(\|z\|^5).
\]
For a scalar argument $s\in[-1,1]$, the Taylor expansion at $s=0$ is
\[
g_\beta(s)=\beta s+b_3s^3+O(s^5),
\qquad
b_3=-\frac{\beta(\beta^2-1)}3,
\]
the cubic graph is
\begin{equation}
h_{3,G}(z)
=-\alpha b_3d\int_0^\infty e^{-\lambda_ut}
u^\top\left[(e^{A_st}z)^{\circ3}\right]dt,
\label{eq:h3-integral}
\end{equation}
where $A_s:u^\perp\to u^\perp$ is the stable transverse restriction of the
Jacobian at the origin, and the power $x^{\circ3}$ means componentwise cubing.
In terms of complete
eigenspace projectors,
\begin{equation}
h_{3,G}(z)
=-\alpha b_3d
\sum_{\lambda_a,\lambda_b,\lambda_c>0}
\frac{
u^\top\left[
(P_{\lambda_a} z)\circprod(P_{\lambda_b} z)\circprod(P_{\lambda_c} z)
\right]
}{\lambda_u-\sigma_{\lambda_a}-\sigma_{\lambda_b}-\sigma_{\lambda_c}}.
\label{eq:h3-projectors}
\end{equation}
The sum runs over the distinct positive Laplacian eigenvalues
$\lambda_a,\lambda_b,\lambda_c\in\operatorname{spec}(L)\setminus\{0\}$, and $\circprod$
denotes componentwise multiplication. This formula is independent of the choice of basis inside a repeated
eigenspace. The eigenvalues fix its denominators, but the projectors and the
labelled pulse fix its nonlinear overlap numerators.

For balanced directions $z,v\in u^\perp$ and a real scale
$\varepsilon\to0$, define the scalar oriented response of the decision
surface to a pulse in direction $v$ by
\begin{equation}
\Gamma_G(\varepsilon;z,v)
=-Dh_G(\varepsilon z)v
=\varepsilon^2K_G(z,v)+O(\varepsilon^4),
\qquad K_G(z,v):=-Dh_{3,G}(z)v.
\label{eq:local-crossing}
\end{equation}
Thus the first response that depends on the graph comes from the curvature of
the basin boundary. The leading coefficient contains overlaps involving the
spectral projectors in
\eqref{eq:h3-projectors}. The remainder in \eqref{eq:local-crossing} concerns
the derivative of the boundary; a finite pulse also has a displacement
remainder.

\begin{proposition}[Local basin crossings with opposite signs on a cospectral pair]
\label{prop:local-branch}
For the registered vectors in \eqref{eq:center-source}, the commonly labelled
rook and Shrikhande graphs $\Rook$ and $\Shri$ on $\mathbb Z_4^2$, defined in
\cref{sec:witness}, satisfy, with $K_G(\tau)=K_G(z_\tau,v_T)$,
\begin{equation}
K_{\Rook}(\tau)>\frac{23}{2500},
\qquad
K_{\Shri}(\tau)<-\frac{23}{2500},
\qquad
0<\tau\leq\frac1{100}.
\label{eq:K-signs}
\end{equation}
With the registered scalar parameters
$\alpha=\kappa=1$, $d=6$ and $\beta=6/5$, the initial conditions
\[
o_G^\pm(0)=\varepsilon z_\tau\pm\delta v_T,
\qquad
0<\varepsilon\leq\frac1{2000},
\qquad
0<\delta\leq\frac{\varepsilon}{1000},
\]
use $\varepsilon$ as the amplitude of the background center and $\delta$ as the
magnitude of the positive order pulse; the displayed $\pm$ sign selects its
direction.
They converge to
\begin{equation}
\begin{array}{c@{\qquad}cc}
\toprule
G&+\delta v_T&-\delta v_T\\
\midrule
\Rook&+q_\star\one&-q_\star\one\\
\Shri&-q_\star\one&+q_\star\one\\
\bottomrule
\end{array}
\label{tab:local-winners}
\end{equation}
where $q_\star\in(0,1)$ is the unique positive root of
$7q=6g_{6/5}(q)$.
\end{proposition}

\begin{proof}[Proof outline]
Let $P_T\in\mathbb R^{16\times16}$ be the permutation matrix induced by
translation of the labelled vertices by $(2,2)$ modulo $4$. It is an
automorphism of both graphs. The signed
symmetry $\mathcal S=-P_T$ satisfies
$\mathcal Sz_\tau=z_\tau$, $\mathcal Sv_T=-v_T$ and
$\mathcal S\one=-\one$. It pins $\varepsilon z_\tau$ to the local stable manifold.
Exact projector arithmetic in \eqref{eq:h3-projectors} gives
\eqref{eq:K-signs}. Explicit Lyapunov--Perron bounds and bounds on the finite
displacement preserve the four strict signs relative to the decision surface
over the stated nonzero range.
Here a strict positive or negative orthant consists of vectors whose
coordinates are all strictly positive or all strictly negative. The two sides
enter opposite strict orthants. Cooperativity and scalar comparison then give the
consensus states in \eqref{tab:local-winners} \cite{Smith1995}. Full rational inequalities are
provided in the Supplementary Information.
\end{proof}

\Cref{prop:local-branch} is the analytic mechanism result. The registered
$0.01$ channel pulse is larger than this proposition allows, so its outcome is
certified directly in the next section.

\section{Explicit opposite winners on a cospectral pair}
\label{sec:witness}

\subsection{Graphs and registered finite kick}

The rook and Shrikhande graphs are a controlled laboratory for separating an
eigenvalue list from the spatial geometry of its eigenspaces. They have the
same size, degree and complete Laplacian spectrum, but different edge sets.
Keeping the vertex labels fixed lets us apply exactly the same background
state and exactly the same four-node order schedule to both graphs.

On $V=\mathbb Z_4^2$, two vertices are adjacent in the rook graph $\Rook$ if
they share a row or column. In the Shrikhande graph $\Shri$, two vertices are
adjacent when their modular difference is one of
\[
(\pm1,0),\qquad(0,\pm1),\qquad(1,1),\qquad(-1,-1).
\]
Both graphs are connected and 6-regular, and
\begin{equation}
\operatorname{spec}(L_{\Rook})
=\operatorname{spec}(L_{\Shri})
=\{0^{(1)},4^{(6)},8^{(9)}\}.
\label{eq:cospectral}
\end{equation}
They are nonisomorphic because their clique numbers are four and three.

The pulse arrays in \eqref{eq:p-vector} were designed under three constraints.
They have zero sum, the four message nodes have zero background opinion, and
translation by $(2,2)$ supplies the signed symmetry used to place the unpulsed
center on the decision surface. The deformation $p$ also prevents the witness
from collapsing to a smaller signed coordinate quotient. For
$0<\tau\leq0.01$, exact signed equitable refinement separates all 16
coordinates on both graphs. The unequal clique numbers exclude a full
signed-permutation conjugacy. Thus the construction lies in no proper
signed-coordinate quotient, and the two graph systems are not related by a
signed-coordinate relabelling. It remains an engineered existence witness
rather than a genericity claim; the exact refinement certificate is given in
the Supplementary Information.

The registered kick uses \eqref{eq:registered-source}. Its center
$(4/5)z_{1/100}$ is nonhomogeneous. Along the subsequent center trajectory,
the Jacobian is node dependent and generally does not commute with the
Laplacian. The homogeneous propagator \eqref{eq:projector-response} and the
local expansion \eqref{eq:local-crossing} therefore do not prove the outcome
at this point.

\medskip
\noindent\textbf{Certified bound at finite amplitude.}\par
For the registered initial conditions
$o_G^{(\eta)}(0)=(4/5)z_{1/100}+\eta(1/100)v_T$, interval integration with
outward rounding of \eqref{eq:registered-flow} to $T=3$ gives, in the positive
arm,
\begin{equation}
\min_i o_{\Rook,i}^{(+)}(3)>2.142125325208\times10^{-5},
\qquad
\max_i o_{\Shri,i}^{(+)}(3)<-2.142758397374\times10^{-5}.
\label{eq:interval-orthants}
\end{equation}
The signed symmetry supplies the reversed inequalities in the negative arm.
The enclosures include roundoff and truncation error in floating point
arithmetic. A validated Taylor map of degree four with step size $10^{-3}$ encloses the four
trajectories for 3000 steps. Exact integer sign tests validate the algebraic
evaluation of $g_{6/5}$, and an outward error recurrence encloses both Taylor
remainder and numerical roundoff. Supplementary Lemma~S6.1 gives the formal
statement, complete construction and constants that can be checked by machine.

Here a strict positive or negative orthant means that all 16 node opinions
have the same strict sign. The values of order $10^{-5}$ are certified margins
at the finite time $T=3$. They are neither estimates of the final consensus
magnitude nor statistical error bars. Once all nodes have the same strict
sign, cooperativity keeps the trajectory in that sign region and scalar
comparison determines its limiting consensus.

\begin{theorem}[Same local order effect, opposite collective limits]
\label{thm:opposite-winners}
Let $G$ be $\Rook$ or $\Shri$ with the common vertex ordering of
$\mathbb Z_4^2$. Use the common pre-kick density states, the crossed atomic
schedule and the registered values in \eqref{eq:registered-source}. After the
kick, evolve every node by \eqref{eq:density-flow} with
$W=A_G$, $\alpha=\kappa=1$ and $\beta=6/5$. Then
\begin{equation}
\begin{array}{c@{\qquad}cc}
\toprule
&\eta=+1&\eta=-1\\
\midrule
\Rook&+q_\star\one&-q_\star\one\\
\Shri&-q_\star\one&+q_\star\one\\
\bottomrule
\end{array}
\label{tab:moderate-winners}
\end{equation}
gives the limiting opinion vectors, where
\begin{equation}
q_\star\in(0,1),
\qquad q_\star\approx0.4338978338,
\qquad
7q_\star=6g_{6/5}(q_\star).
\label{eq:qstar}
\end{equation}
In terms of the full density matrices, every node converges to
\begin{equation}
\rho_i(t)\longrightarrow
\frac12\left(I_2\pm q_\star\sigma_z\right),
\label{eq:density-winners}
\end{equation}
with the signs in \eqref{tab:moderate-winners}. Thus the same finite channels,
nodewise preparations, labels and crossed-order schedule select opposite
macroscopic opinions on two networks with the same complete Laplacian
eigenvalue multiset.
\end{theorem}

\begin{proof}
By \cref{thm:local-source}, the atomic crossed schedule generates exactly the
state after the kick in \eqref{eq:postkick}. By \cref{thm:network-flow}, its expressed
opinion follows \eqref{eq:registered-flow} exactly, and its contextual
coordinate decays as $e^{-7t}$.

The certified bounds in \eqref{eq:interval-orthants} give the strict positive
and negative orthants, and symmetry gives the reversed arms.
For a trajectory in the strict positive orthant, scalar minimum and maximum
comparison squeeze every coordinate between positive solutions of
$\dot q=-7q+6g_{6/5}(q)$, both of which converge to the unique positive
equilibrium $q_\star$. Indeed, elementary differentiation shows that
$q\mapsto g_{6/5}(q)/q$ decreases strictly from $6/5$ to $1$ on $(0,1)$, so it
crosses $7/6$ exactly once. The negative case follows by odd symmetry. This proves
\eqref{tab:moderate-winners}. Finally, the semiconjugacy and the decay of
$c_i$ lift the scalar limits to \eqref{eq:density-winners}.
\end{proof}

\begin{figure}[!t]
\centering
\IfFileExists{figures/fig2_witness.pdf}{%
  \includegraphics[width=\linewidth]{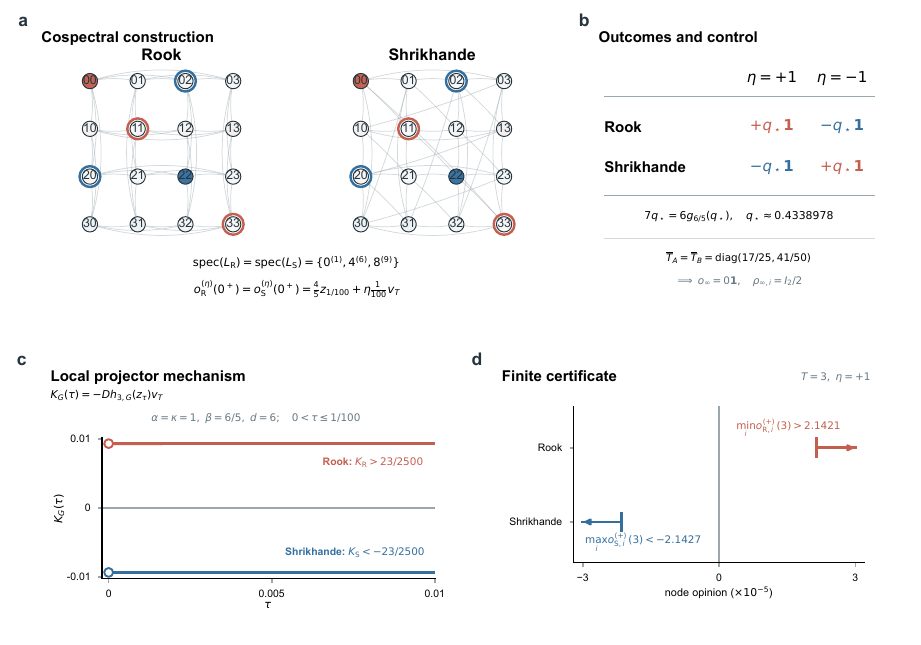}%
}{%
  \fbox{\parbox[c][0.27\textheight][c]{0.94\linewidth}{\centering
  Placeholder for the cospectral witness figure.\\[0.5em]
  Rook and Shrikhande graphs, labelled pulse arrays, projector signs,
  interval margins and the opposite-winner table.}}
}
\caption{\textbf{A cospectral construction, two parallel proof branches and
their outcomes.}
Panel a shows the rook and Shrikhande graphs with the same $N=16$ vertex
labels, degree $d=6$, complete Laplacian spectrum and registered input.
The node interiors encode the registered background
$(4/5)z_{1/100}\in\mathbb R^{16}$, while
the rings encode the balanced source $v_T\in\{-1,0,1\}^{16}$. The displayed
rings correspond to the arm $\eta=+1$; $\eta=-1$ reverses their signs.
Panel b records the resulting limits, where $q_\star$ is the unique positive
root shown in the panel, and the separate exact control for the parity-twirled
Bloch maps $\overline T_A$ and $\overline T_B$.
Panels c and d are separate proof branches that use the same graph pair, labels
and source directions at different amplitudes. Panel c gives the local
small-amplitude mechanism; the complete eigenspace projectors $P_\lambda$ enter
$h_{3,G}$ through \eqref{eq:h3-projectors}. In the local initial condition
$o_G^\pm(0)=\varepsilon z_\tau\pm\delta v_T$, $\varepsilon$ and $\delta$ are
the background and pulse amplitudes. Its theorem is restricted to
$0<\varepsilon\leq1/2000$ and $0<\delta\leq\varepsilon/1000$. Panel d gives an
independent deterministic certificate at the registered background scale
$\varrho=4/5$, center deformation $\tau=1/100$ and pulse amplitude
$\mu=1/100$. Its bounds are computed with outward rounding. The rays are
one-sided enclosures rather than statistical error bars; the full precision
bounds appear in
\eqref{eq:interval-orthants}. Signed symmetry supplies the arm $\eta=-1$.
Under the control in panel b, both arms on both graphs converge to $I_2/2$ at
every node.}
\label{fig:witness}
\end{figure}
\FloatBarrier

\subsection{Controls and local robustness of the mechanism}

The local parity twirl in \eqref{eq:twirl} removes the opinion pulse that
depends on the arm while leaving the preparation unchanged. Both arms then start the
network flow from $(4/5)z_{1/100}$. The signed translation symmetry places this
center in an invariant zero-mean subspace. The identity
$A_G^2=4I_N+2\one\one^\top$ and the bound
$0\leq g_{6/5}'\leq6/5$ make that transverse subspace contract, so the center
converges to zero, and context also decays.

\paragraph{Network-level parity control.}
If each local channel is replaced by its parity twirl while the preparation is
held fixed, both arms converge on both graphs to $I_2/2$ at every node. The
macroscopic arm difference is zero.

The four initial conditions in \cref{thm:opposite-winners} lie in open basins
of asymptotically stable consensus states. Finite channel composition and the
common preparation depend continuously on the two channel axes, their two
partial-dephasing fractions and the preparation vector. Taking the finite
intersection of the four basin preimages, one for each graph and arm, gives an
open neighborhood common to all four cases.

To specify the relevant topology, let
$\mathring{\mathbb B}_2=\{r\in\mathbb R^2:\|r\|_2<1\}$ be the open Bloch disk.
For an axis $n\in S^1$, a dephasing fraction $\chi\in(0,1)$ and a preparation
$r\in\mathring{\mathbb B}_2$, define
$\mathcal E_{n,\chi}=(1-\chi)\Id+\chi\mathcal D_n$, where $\mathcal D_n$ is
given by the projector formula preceding \eqref{eq:partial-dephasing} with axis $n$.
The channel parameter
space is
\[
\mathcal P=(S^1)^2\times(0,1)^2\times\mathring{\mathbb B}_2,
\]
whose coordinates are the two channel axes, two dephasing fractions and the
common preparation at the active nodes. Its registered point is
$(n_A,n_B,1/2,1/2,r_{\rm pre})\in\mathcal P$.

\begin{corollary}[Open parameter neighborhood]
With the graphs, labelled schedule and network flow fixed, there is an open
neighborhood, relative to $\mathcal P$, of the
registered channel axes, channel strengths and common preparation at the active nodes
for which all four entries of
\eqref{tab:moderate-winners} are unchanged. After shrinking the neighborhood
if necessary, the two channels remain noncommuting and their ordered
compositions retain a nonzero opinion contrast on the common preparation.
\end{corollary}

The corollary establishes persistence of the outcome with opposite winners and
of the local order contrast in some open neighborhood of the registered point.
It does not quantify that neighborhood; doing so would require an additional
certificate.

\section{Discussion}

The main topological insight is finer than the statement that network structure
matters. Nor is it merely that eigenvectors contain information absent from
eigenvalues. A specific overlap involving eigenspace projectors enters the
leading curvature of the decision surface for a fixed labelled state and
pulse. The rook and Shrikhande graphs share every Laplacian eigenvalue, but this
coefficient has opposite signs. The registered example then establishes the
finite outcome independently by following the nonlinear trajectories with
validated enclosures.

The model is intentionally modular. Noncommuting channels provide one
initialization whose result depends on order, and a common autonomous nonlinear
flow then turns the resulting pulse into a basin choice. The channels act only
at initialization; subsequent edge dynamics follows the common flow. The
density matrix formulation still serves three structural roles. It constrains
every local preparation and update to a valid state space, exposes the pathway
from context to opinion and its controls, and supplies a network equation for
the complete state whose opinion coordinate is exact. Context does not feed
back into opinion after the kick, so the later feedback from senders is
classical.

The construction is operational and does not treat either a human or an LLM as
a physical quantum system. General representations based on classical
instruments can also reproduce broad classes of sequential decision data
\cite{TullOzawa2026}. The contribution is a closed mathematical example in
which a finite noncommuting update pair, a valid flow of complete states and a
labelled cospectral comparison can all be checked exactly or with rigorous
numerical bounds.

The next empirical test is to estimate a separate update map for each message
type from
LLM responses under reproducible preparation and readout prompts. Existing LLM
agent studies already show consensus, fragmentation, history dependence,
asymmetric persuasion and dependence on presentation order
\cite{ChuangEtAl2024,CisnerosVelarde2024,CauEtAl2025,
StengelEskinEtAl2025}. The cited studies do not identify the two channels assumed here.
A useful calibration would therefore test predictions for both $AB$ and $BA$
without refitting, quantify uncertainty in the estimated pulse, and compare that
uncertainty with a certified basin margin. The current paper establishes the
theoretical target for that experiment. Its graph pair and pulse remain an
engineered existence witness. The claim does not extend to cospectral networks
universally or generically.

\section{Conclusion}

The same two finite updates produce opinions $+0.01$ and $-0.01$ when their
order is reversed. A balanced four-node placement of this effect selects
opposite consensuses on two networks with the same complete Laplacian
eigenvalue list. A particular overlap of spectral projectors determines the
leading local response of the basin boundary, and a rigorous interval
certificate independently proves the registered outcomes at finite amplitude.
Within the declared model, the result gives a precise route from a single
balanced pulse of message order to collective selection, while keeping the
quantum-like update, classical network feedback and two proof regimes
explicitly separated.

\section*{Data and code availability}

The paper uses no empirical data. The joint algebraic verification and the
verification at finite amplitude are included with the arXiv source as
ancillary code. From the submission source root, run
\begin{center}
\texttt{python3 -B anc/paper1\_rebit\_theory/verify\_rebit\_closure\_v0\_1.py
--with-winner}.
\end{center}
The Supplementary Information records the exact rational identities,
stable-manifold bounds, outward interval construction and reproduction
commands.

{\small
\bibliographystyle{unsrt}
\bibliography{references}
}

\end{document}


\maketitle

This Supplement gives the exact local channel algebra, the network equation
that preserves density matrices, and the two independent winner certificates
used in the main text. Section~S5 uses a small-radius calculation to explain
the role of complete Laplacian eigenspace projectors. Section~S6 independently
proves the certificate at moderate amplitude by outward interval enclosures.

\section{Partial-dephasing channels and exact order algebra}
\label{sec:channels}

Let $\mathbb M_2=\mathbb C^{2\times2}$ be the complex $2\times2$ matrix
algebra, let $I_2\in\mathbb R^{2\times2}$ be its identity matrix, and let
$\sigma_x=\left(\begin{smallmatrix}0&1\\1&0\end{smallmatrix}\right)$ and
$\sigma_z=\left(\begin{smallmatrix}1&0\\0&-1\end{smallmatrix}\right)$ be the
two real Pauli matrices. Write $\Tr$ for matrix trace. The rebit state space and
its Bloch disk are
\[
 \mathcal D_{\rm R}=\{\rho\in\mathbb R^{2\times2}:\rho=\rho^\top,
 \ \rho\succeq0,\ \Tr\rho=1\}\subset\mathbb M_2,
 \qquad
 \mathbb B_2=\{r\in\mathbb R^2:\|r\|_2\leq1\}.
\]
At each node, the operational state $\rho\in\mathcal D_{\rm R}$ is
\begin{equation}
 \rho=\frac12\left(I_2+c\sigma_x+o\sigma_z\right),
 \qquad r=(c,o)^\top\in\mathbb B_2,
 \qquad c,o\in\mathbb R.
 \label{eq:rebit-state}
\end{equation}
The readout \(o=\Tr(\sigma_z\rho)\in[-1,1]\) is the expressed opinion, and
\(c=\Tr(\sigma_x\rho)\in[-1,1]\) is the auxiliary context coordinate. Choose
two unit axes $n_A,n_B\in S^1:=\{n\in\mathbb R^2:\|n\|_2=1\}$ and, for a
message label $\ell\in\{A,B\}$, set
\begin{equation}
 n_A=\left(\frac35,\frac45\right)^\top,
 \qquad
 n_B=\left(\frac35,-\frac45\right)^\top,
 \qquad P_\ell=n_\ell n_\ell^\top\in\mathbb R^{2\times2}.
 \label{eq:axes}
\end{equation}
The vectors $n_A$ and $n_B$ are the two channel axes. Define
\begin{equation}
 \Pi_\ell^\pm=\frac12\left[I_2\pm
 n_\ell\mathbin{\cdot}(\sigma_x,\sigma_z)\right],
 \qquad
 \Dens_\ell(X)=\Pi_\ell^+X\Pi_\ell^+
                  +\Pi_\ell^-X\Pi_\ell^- ,
 \label{eq:complete-dephasing}
\end{equation}
where $n\mathbin{\cdot}(\sigma_x,\sigma_z)=n_1\sigma_x+n_2\sigma_z$ for
$n=(n_1,n_2)^\top\in\mathbb R^2$, and
$\Dens_\ell:\mathbb M_2\to\mathbb M_2$ is the complete-dephasing
superoperator. Writing $\Id:\mathbb M_2\to\mathbb M_2$ for the identity
superoperator, use the finite partial dephasing
\begin{equation}
 \mathcal E_\ell=\frac12\Id+\frac12\Dens_\ell.
 \label{eq:partial-dephasing}
\end{equation}

\paragraph{Exact order pulse from channels of finite strength.}
The maps \(\mathcal E_A\) and \(\mathcal E_B\) are unital CPTP maps that
preserve the real Bloch plane.  Their Bloch matrices are
\begin{equation}
 T_A=
 \begin{pmatrix}
 17/25&6/25\\[1mm]6/25&41/50
 \end{pmatrix},
 \qquad
 T_B=
 \begin{pmatrix}
 17/25&-6/25\\[1mm]-6/25&41/50
 \end{pmatrix}.
 \label{eq:TA-TB}
\end{equation}
They have the same eigenvalues \(1\) and \(1/2\), but
\begin{equation}
 T_BT_A-T_AT_B
 =\frac{42}{625}
 \begin{pmatrix}0&-1\\1&0\end{pmatrix}\neq0.
 \label{eq:commutator}
\end{equation}
For the common preparation
\begin{equation}
 r_{\rm pre}=\left(\frac{25}{84},0\right)^\top,
 \label{eq:preparation}
\end{equation}
the convention \(AB=\mathcal E_B\circ\mathcal E_A\) gives
\begin{equation}
 T_BT_Ar_{\rm pre}
 =\left(\frac{253}{2100},\frac1{100}\right)^\top,
 \qquad
 T_AT_Br_{\rm pre}
 =\left(\frac{253}{2100},-\frac1{100}\right)^\top.
 \label{eq:ordered-outputs}
\end{equation}
Thus the two orders have the same contextual endpoint and opposite opinion
endpoints.

\begin{proof}[CPTP verification and exact algebra]
The operators
\(\{I_2/\sqrt2,\Pi_\ell^+/\sqrt2,\Pi_\ell^-/\sqrt2\}\)
form a Kraus representation of \(\mathcal E_\ell\), since the sum of their
adjoint products is \(I_2\). Complete dephasing projects the Bloch vector onto
\(n_\ell\), so \(T_\ell=(I_2+P_\ell)/2\). Substitution of
\eqref{eq:axes} gives \eqref{eq:TA-TB}, and exact multiplication gives
\eqref{eq:commutator} and \eqref{eq:ordered-outputs}.  The preparation norm is
\(25/84<1\), while either output has squared norm
\begin{equation}
 \left(\frac{253}{2100}\right)^2+\left(\frac1{100}\right)^2
 =\frac{1289}{88200}<1.
 \label{eq:output-norm}
\end{equation}
Hence every displayed state lies strictly inside the Bloch disk.  The
transverse contraction \(1/2\) corresponds to the finite dephasing strength
\(\log2\) in a semigroup parameterization.
\end{proof}

Both channels are unital and strength matched, so no affine translation
creates the sign in \eqref{eq:ordered-outputs}. Also, for every $n\in S^1$,
\(\Dens_n=\Dens_{-n}\), which prevents the written sign of an axis from
encoding a preferred opinion.

\section{Population and parity-twirl controls}
\label{sec:controls}

Let $\Delta_Z:\mathbb R^2\to\mathbb R^2$ be the population projection
\(\Delta_Z(c,o)=(0,o)\), the Bloch action of complete dephasing in the
declared opinion basis.  Restricting the whole local model to its population
sector means applying
\(\Delta_ZT_\ell\Delta_Z\), including the initial projection.  Exact matrix
multiplication gives
\begin{equation}
 \Delta_ZT_A\Delta_Z=\Delta_ZT_B\Delta_Z
 =\begin{pmatrix}0&0\\0&41/50\end{pmatrix},
 \qquad
 \Delta_Zr_{\rm pre}=0.
 \label{eq:population-control}
\end{equation}
Both ordered population-sector compositions therefore end at zero and have
zero order contrast.  A projection inserted only after the first full update
is not this restriction, because that first update has already transferred
information from the contextual preparation into the opinion coordinate.

A second control preserves the original preparation. Let
\begin{equation}
 \mathcal Z:\mathbb M_2\to\mathbb M_2,
 \qquad \mathcal Z(X)=\sigma_zX\sigma_z,
 \qquad
 \overline{\mathcal E}_\ell
 =\frac12\left(\mathcal E_\ell+
 \mathcal Z\circ\mathcal E_\ell\circ\mathcal Z\right).
 \label{eq:parity-twirl}
\end{equation}
Conjugation by \(\sigma_z\) sends \((c,o)\) to \((-c,o)\).  Consequently,
\begin{equation}
 \overline T_A=\overline T_B
 =\operatorname{diag}\left(\frac{17}{25},\frac{41}{50}\right),
 \label{eq:twirled-matrix}
\end{equation}
and both orders applied to the unchanged preparation give
\begin{equation}
 \overline T_B\overline T_Ar_{\rm pre}
 =\overline T_A\overline T_Br_{\rm pre}
 =\left(\frac{289}{2100},0\right)^\top.
 \label{eq:twirled-output}
\end{equation}
The full model's contrast therefore requires the intermediate contextual
coordinate relative to the fixed preparation and readout basis. The conclusion
is restricted to this basis; nonclassicality that is independent of basis and
comparisons with general classical models with latent states lie outside these
controls.

\section{Nonlinear network flow that preserves density matrices}
\label{sec:flow}

For \(s\in\{-1,+1\}\), define
\begin{equation}
 \Omega_s=\frac12(I_2+s\sigma_z)\in\mathcal D_{\rm R},
 \qquad \Omega_0=\frac {I_2}2\in\mathcal D_{\rm R}.
 \label{eq:reset-generator}
\end{equation}
For \(s\in\{-1,0,+1\}\), let
$\mathcal L_s:\mathbb M_2\to\mathbb M_2$ be the linear replacement generator
\(\mathcal L_s(X)=\Tr(X)\Omega_s-X\) for $X\in\mathbb M_2$.
Each \(\mathcal L_s\) generates the replacement semigroup
\(X\mapsto e^{-t}X+(1-e^{-t})\Tr(X)\Omega_s\) for $t\in[0,\infty)$, a special case of a GKLS
generator~\cite{GoriniEtAl1976,Lindblad1976}.
For a fixed parameter \(\beta\in[1,\infty)\) that controls stance sharpening, define the
probability kernel
$q_\beta:\{-1,+1\}\times[-1,1]\to[0,1]$ by
\begin{equation}
 q_\beta(s\mid o)
 =\frac{(1+so)^\beta}{(1+o)^\beta+(1-o)^\beta},
 \qquad -1\leq o\leq1,
 \label{eq:sender-probability}
\end{equation}
so that $\sum_{s=\pm1}q_\beta(s\mid o)=1$. Its mean-stance map
$g_\beta:[-1,1]\to[-1,1]$ is
\begin{equation}
 g_\beta(o)=\sum_{s=\pm1}s q_\beta(s\mid o)
 =\tanh\!\left(\beta\operatorname{artanh}o\right),
 \qquad g_\beta(\pm1)=\pm1.
 \label{eq:transmitted-stance}
\end{equation}
The endpoint values in \eqref{eq:transmitted-stance} are continuous
extensions.

Let $N\in\mathbb N$ be the number of nodes, use
$i,j\in\{1,\ldots,N\}$ as node indices, and let
$W=(w_{ij})\in\mathbb R_{\geq0}^{N\times N}$ be a finite nonnegative influence
matrix. Here $w_{ij}$ is the influence weight from node $j$ to node $i$ and
\(d_i=\sum_jw_{ij}\in[0,\infty)\) is the weighted degree. Let
\(\alpha,\kappa\in[0,\infty)\) be the rate of influence from neighbors and the
rate of relaxation to the neutral state. For continuous time
$t\in[0,\infty)$, the propagation after the kick is
\begin{equation}
 \dot\rho_i
 =\alpha\sum_jw_{ij}\sum_{s=\pm1}
 q_\beta(s\mid o_j)\mathcal L_s(\rho_i)
 +\kappa\mathcal L_0(\rho_i).
 \label{eq:N1}
\end{equation}
The state dependence in \(q_\beta\) makes \eqref{eq:N1} a nonlinear global
flow, so the linear CPTP classification does not apply. At each fixed global
state, the receiver dynamics is a nonnegative mixture of reset generators.
The result that preserves states does not require $W$ to be symmetric; the main
text later specializes to the undirected case $W=W^\top$ with $G=(V,W)$.

\paragraph{Global state preservation and exact Bloch semiconjugacy.}
Suppose \(w_{ij}\geq0\), \(\alpha,\kappa\geq0\), and \(\beta\geq1\).
For every collection \((\rho_1,\ldots,\rho_N)\) of rebit density matrices in
the Cartesian state space, \eqref{eq:N1} has a unique global forward solution.
Trace, Hermiticity, the real Bloch plane, and positivity are preserved at
every node. Its Bloch coordinates obey
\begin{align}
 \dot c_i&=-(\alpha d_i+\kappa)c_i,
 \label{eq:context-decay}\\
 \dot o_i&=\alpha\left[-d_i o_i+\sum_jw_{ij}g_\beta(o_j)\right]
            -\kappa o_i.
 \label{eq:opinion-flow}
\end{align}
Define the flow of the complete state
$\Phi_W^t:\mathcal D_{\rm R}^N\to\mathcal D_{\rm R}^N$, the opinion flow
$\phi_W^t:[-1,1]^N\to[-1,1]^N$ of \eqref{eq:opinion-flow}, and the readout
projection $\pi_o:\mathcal D_{\rm R}^N\to[-1,1]^N$ by
\[
 \pi_o(\rho_1,\ldots,\rho_N)
 =\bigl(\Tr(\sigma_z\rho_i)\bigr)_{i=1}^N.
\]
For a symmetric $W$ encoded by $G$, the main-text notation is
$\Phi_G^t:=\Phi_W^t$ and $\phi_G^t:=\phi_W^t$.
Then
\begin{equation}
\pi_o\circ\Phi_W^t=\phi_W^t\circ\pi_o.
\label{eq:semiconjugacy}
\end{equation}

\begin{proof}[Verification]
Put \(\gamma_i=\alpha d_i+\kappa\in[0,\infty)\), the total replacement rate
at node $i$. When \(\gamma_i>0\), write
$\boldsymbol\rho=(\rho_j)_{j=1}^N\in\mathcal D_{\rm R}^N$ for the network state
and define $\Theta_i:\mathcal D_{\rm R}^N\to\mathcal D_{\rm R}$ by
\begin{equation}
 \Theta_i(\boldsymbol\rho)=\frac{
 \alpha\sum_jw_{ij}\sum_{s=\pm1} q_\beta(s\mid o_j)\Omega_s
 +\kappa\Omega_0}{\gamma_i}.
 \label{eq:theta-target}
\end{equation}
The coefficients in \eqref{eq:theta-target} are nonnegative and sum to one,
so \(\Theta_i(\boldsymbol\rho)\) is a density matrix. Equation \eqref{eq:N1} becomes
\(\dot\rho_i=\gamma_i(\Theta_i(\boldsymbol\rho)-\rho_i)\), and variation of constants
gives
\begin{equation}
 \rho_i(t)=e^{-\gamma_i t}\rho_i(0)
 +\int_0^t\gamma_i e^{-\gamma_i(t-s)}
 \Theta_i(\boldsymbol\rho(s))\,ds.
 \label{eq:convex-variation}
\end{equation}
The total scalar weight in \eqref{eq:convex-variation} is one.  If
\(\gamma_i=0\), all rates at node \(i\) vanish and \(\dot\rho_i=0\).

For fixed \(\beta\geq1\), \(q_\beta\) is Lipschitz on \([-1,1]\).  Take a
Lipschitz extension to an open neighborhood of that interval.  Standard local
existence and uniqueness apply to the extended vector field.  Up to any first
exit from the set of density states, \eqref{eq:convex-variation} remains a convex
combination of density matrices, which rules out such an exit.  Compactness
then yields global forward existence.  Taking traces against \(\sigma_x\) and
\(\sigma_z\) gives \eqref{eq:context-decay} and
\eqref{eq:opinion-flow}, which proves \eqref{eq:semiconjugacy}.
\end{proof}

Let $D=\operatorname{diag}(d_1,\ldots,d_N)\in\mathbb R^{N\times N}$ be the
degree matrix, $L=D-W\in\mathbb R^{N\times N}$ the network Laplacian,
$I_N\in\mathbb R^{N\times N}$ the identity matrix, and
$o=(o_i)_{i=1}^N\in[-1,1]^N$ the opinion vector. With $g_\beta$ acting
componentwise,
the opinion equation can also be written
\begin{equation}
 \dot o=-(\alpha L+\kappa I_N)o+\alpha W\bigl(g_\beta(o)-o\bigr).
 \label{eq:laplacian-form}
\end{equation}
For either registered 6-regular graph $G$ introduced in Section~S4, let
$A_G\in\{0,1\}^{16\times16}$ be its adjacency matrix. Then
\(\alpha=\kappa=1\) and
\(\beta=6/5\) reduce the exact opinion coordinate to
\begin{equation}
 \dot o=-7o+A_Gg_{6/5}(o),
 \qquad c_i(t)=e^{-7t}c_i(0).
 \label{eq:registered-flow}
\end{equation}
The complete model applies the atomic local channel map at \(t=0\) and the
autonomous flow \eqref{eq:N1} for \(t>0\). Each atomic block contains exactly
the two channel updates; the common network flow begins only after the block
and alone propagates along edges.

\section{Cospectral graphs and the crossed order pulse}
\label{sec:graphs}

Use the common labelled vertex set \(V=\mathbb Z_4^2\), with all coordinates
taken modulo four, so $N=|V|=16$. Let $G\in\{\Rook,\Shri\}$ denote the rook or
Shrikhande graph, let $A_G\in\{0,1\}^{16\times16}$ be its adjacency matrix,
let $I_{16}$ be the identity on $\mathbb R^{16}$, and write
$\one=(1,\ldots,1)^\top\in\mathbb R^{16}$ and
$J=\one\one^\top\in\mathbb R^{16\times16}$. The graph Laplacian is
$L_G=6I_{16}-A_G$. The rook graph joins distinct vertices in the same row or
column. The Shrikhande graph uses the steps
\begin{equation}
 (\pm1,0),\quad(0,\pm1),\quad(1,1),\quad(-1,-1).
 \label{eq:shrikhande-steps}
\end{equation}
Their adjacency matrices satisfy, by exact integer arithmetic,
\begin{equation}
 A_G\one=6\one,
 \qquad A_G^2=4I_{16}+2J.
 \label{eq:strongly-regular-identity}
\end{equation}
It follows that
\begin{equation}
 \operatorname{spec}(A_G)=\{6^{(1)},2^{(6)},(-2)^{(9)}\},
 \qquad
 \operatorname{spec}(L_G)=\{0^{(1)},4^{(6)},8^{(9)}\}.
 \label{eq:common-spectrum}
\end{equation}
The graphs are not isomorphic: their clique numbers are four and three,
respectively.  This pair is a standard example of nonisomorphic cospectral
graphs~\cite{vanDamHaemers2003}.

Let \(e_{ij}\in\mathbb R^{16}\) be the standard basis vector at vertex
\((i,j)\in V\), and, for a deformation parameter $\tau\in\mathbb R$, define
\begin{align}
 p={}&(e_{01}-e_{23})+2(e_{03}-e_{21})+3(e_{10}-e_{32})\notag\\
    &+4(e_{12}-e_{30})+5(e_{13}-e_{31}),
 \label{eq:p-vector}\\
 z_\tau={}&e_{00}-e_{22}+\tau p,
 \qquad
 v_T=e_{11}+e_{33}-e_{02}-e_{20}.
 \label{eq:center-source}
\end{align}
Here $p\in\mathbb R^{16}$ specifies the deformation of the center,
$z_\tau\in\mathbb R^{16}$ is the zero-sum background pattern, and
$v_T\in\mathbb R^{16}$ is the signed order pulse supported on four nodes.
The supports of $p$, $e_{00}-e_{22}$ and $v_T$ are pairwise disjoint, and
\begin{equation}
 \|p\|_2^2=110,
 \quad \|z_\tau\|_2^2=2+110\tau^2,
 \quad \|v_T\|_2^2=4,
 \quad z_\tau^\top v_T=0.
 \label{eq:source-norms}
\end{equation}
In particular, every active coordinate in \(\supp(v_T)\) has zero center
opinion. Let $P_T\in\mathbb R^{16\times16}$ be the permutation matrix on node
vectors induced by translation through $(2,2)$ modulo $4$, and define the
signed symmetry \(\mathcal S=-P_T\). Then both graphs obey
\begin{equation}
 \mathcal Sz_\tau=z_\tau,
 \qquad \mathcal Sv_T=-v_T,
 \qquad \mathcal S\one=-\one.
 \label{eq:signed-symmetry}
\end{equation}

\paragraph{Certificate excluding a signed coordinate reduction.}
For every $0<\tau\leq1/100$, the joint signed coordinate signatures of
$(z_\tau,v_T)$ have seven initial cells. Signed equitable refinement by
signatures of neighbor counts has 12 cells after one step and 16 singleton cells
after the second step on both graphs. Hence no proper signed coordinate
polydiagonal contains the registered initial family. A full signed-coordinate
permutation conjugacy between the two graph systems is also impossible.

\begin{proof}[Exact refinement check]
Every coordinate signature is affine in $\tau$. Exact rational comparison
shows that no coordinate signature becomes zero for positive $\tau$, and that
two signatures coincide only at
$1/5,1/4,1/3,1/2,1$, all outside the registered interval. The initial signed
partition is therefore constant on $0<\tau\leq1/100$. Exact integer
refinement by neighbor counts gives the following cell sizes:
\[
 (4,2,2,2,2,2,2),
 \quad
 (2,2,2,2,1,1,1,1,1,1,1,1),
 \quad
 (1,1,1,1,1,1,1,1,1,1,1,1,1,1,1,1)
\]
for both graphs. The final profile is discrete, which excludes a proper signed
coordinate quotient. A full signed-coordinate permutation would conjugate the
two adjacency matrices and hence give a graph isomorphism, which is excluded
by their clique numbers four and three.
\end{proof}

Let $\varrho\in[0,1]$ be the amplitude of the background state,
$\tau\in\mathbb R$ the parameter that deforms its center and
$\mu\in[0,\infty)$ the pulse amplitude. Fix their registered values
\begin{equation}
 \varrho=\frac45,
 \qquad \tau=\frac1{100},
 \qquad \mu=\frac1{100}.
 \label{eq:registered-parameters}
\end{equation}
Before the kick, each active node has the common preparation
\eqref{eq:preparation}; every inactive node has Bloch vector
\((0,\varrho z_{\tau,i})^\top\).  This nodewise preparation profile is the
same on both graphs and in both global arms.  In arm \(\eta\in\{-1,+1\}\), an
active node receives \(AB\) when \(\eta v_{T,i}=+1\) and \(BA\) when
\(\eta v_{T,i}=-1\).  Each arm therefore has two blocks of each order, and every
active node receives exactly the multiset \(\{A,B\}\).

Let
$\mathbf R=(\rho_i)_{i\in V}\in\mathcal D_{\rm R}^{16}$ denote the full
network state, and let
$K_\eta:\mathcal D_{\rm R}^{16}\to\mathcal D_{\rm R}^{16}$ denote the
synchronous atomic map comprising four blocks. No network
evolution, readout, adaptation, or extra update occurs anywhere during that
indivisible kick. Equations
\eqref{eq:ordered-outputs} and \eqref{eq:center-source} give the exact state
after the kick. Write
$o^{(\eta)}(0^+)=(o_i^{(\eta)}(0^+))_{i\in V}\in[-1,1]^{16}$ for its opinion
vector. Then
\begin{equation}
 \mathbf R(0^+)=K_\eta\mathbf R(0^-),
 \qquad
 o^{(\eta)}(0^+)=\varrho z_\tau+\eta\mu v_T,
 \label{eq:post-kick-opinion}
\end{equation}
Here $0^-$ and $0^+$ are the instants immediately before and after the kick.
The autonomous network clock is restarted after the kick, so
$o^{(\eta)}(0):=o^{(\eta)}(0^+)\in[-1,1]^{16}$.
The context coordinates at the restarted time are
\begin{equation}
 c_i(0)=
 \begin{cases}
 253/2100,&i\in\supp(v_T),\\
 0,&i\notin\supp(v_T).
 \end{cases}
 \label{eq:post-kick-context}
\end{equation}
All node states remain inside the Bloch disk.  The graph, labels, preparation
profile, channel pair, channel strengths, active roles, and two crossed
schedules are fixed before the graph-dependent flow begins.
The opinion perturbation \(\eta\mu v_T\) is balanced, whereas the nonzero
contextual component in \eqref{eq:post-kick-context} is identical in both arms.

\section{Local projector mechanism with a bound for finite pulses}
\label{sec:local}

This section proves the small-amplitude theorem corresponding to Proposition~1
in the main text; Section~S6 gives the independent certificate at moderate
amplitude. We reserve $\rho_i$ for density matrices. In this section,
$\varepsilon>0$ scales the background center and $\delta>0$ is the magnitude
of the positive order pulse; a separate sign selects its direction. Write
$x=o\in[-1,1]^{16}$ for the local opinion vector and define
$F_G:[-1,1]^{16}\to\mathbb R^{16}$ by
\begin{equation}
 F_G(x)=-7x+A_Gg(x),
 \qquad
 g(s)=\tanh\!\left(\frac65\operatorname{artanh}s\right)\quad(|s|<1),
 \qquad g(\pm1)=\pm1,
 \qquad u=\frac{\one}{4}.
 \label{eq:local-vector-field}
\end{equation}
Here $g:[-1,1]\to[-1,1]$ uses the displayed continuous endpoint values and
acts componentwise, while
$u\in\mathbb R^{16}$ is the unit consensus vector. Set the consensus and
transverse projectors $P=uu^\top\in\mathbb R^{16\times16}$ and
\(Q=I_{16}-P\). The origin has unstable consensus rate
\(\lambda_u=1/5\), while the Laplacian modes \(4\) and \(8\) have rates
\begin{equation}
 \sigma_4=-\frac{23}{5},
 \qquad \sigma_8=-\frac{47}{5}.
 \label{eq:stable-rates}
\end{equation}
Standard theory of invariant manifolds gives a local stable graph near the origin
\cite{HirschPughShub1977}. Thus, for a neighborhood
$U_G\subset u^\perp:=\{z\in\mathbb R^{16}:u^\top z=0\}$, there is an odd map
$h_G:U_G\to\mathbb R$ with homogeneous cubic term
$h_{3,G}:u^\perp\to\mathbb R$ such that
\begin{equation}
 x=h_G(z)u+z,
 \qquad z\perp u,
 \qquad h_G(z)=h_{3,G}(z)+O(\|z\|_2^5).
 \label{eq:stable-graph}
\end{equation}

\begin{lemma}[Validated Lyapunov--Perron chart]
\label{lem:LP-chart}
Let \(R_{\rm LP}=10^{-3}\) be the Lyapunov--Perron state radius and let
\(\gamma=1\) be the exponential weight. Define the Banach space
\[
 \mathcal X_\gamma=\left\{X\in C([0,\infty),\mathbb R^{16}):
 \|X\|_\gamma:=\sup_{t\geq0}e^{\gamma t}\|X(t)\|_2<\infty\right\}.
\]
For operator-valued trajectories, $\|\cdot\|_\gamma$ denotes the analogous
weighted supremum of the induced Euclidean operator norm.
Write \(F_G(x)=\mathcal L_Gx+N_G(x)\), where
\(\mathcal L_G=-7I_{16}+(6/5)A_G\in\mathbb R^{16\times16}\) is the
linearization at the origin and
$N_G:[-1,1]^{16}\to\mathbb R^{16}$ is the nonlinear remainder
\(N_G(x)=A_G[g(x)-(6/5)x]\). Uniformly for the two graphs and
\(\|x\|_2\leq R_{\rm LP}\),
\begin{equation}
 \|N_G(x)\|_2\leq1.056007\|x\|_2^3,
 \quad
 \|DN_G(x)\|_2\leq3.168035\|x\|_2^2,
 \quad
 \|DN_{5,G}(x)\|_2\leq35\|x\|_2^4,
 \label{eq:nonlinear-majorants}
\end{equation}
where \(N_{5,G}\) is the remainder after the cubic term, restricted here to
the ball $\{x\in\mathbb R^{16}:\|x\|_2\leq R_{\rm LP}\}$. The
Lyapunov--Perron map on \(\mathcal X_\gamma\) is a contraction with constant
\(q_{\rm LP}<2.785\times10^{-6}\). Hence, with the inverse contraction margin
\begin{equation}
 B=(1-q_{\rm LP})^{-1}<1.000003,
 \qquad
 r_{\rm chart}=\frac{R_{\rm LP}}{B},
 \label{eq:chart-radius}
\end{equation}
the stable graph is defined for \(\|z\|_2\leq r_{\rm chart}\) and satisfies
\begin{equation}
 \|Dh_G(z)-Dh_{3,G}(z)\|_2
 \leq C_4\|z\|_2^4,
 \qquad C_4<10.387<11.
 \label{eq:graph-remainder}
\end{equation}
\end{lemma}

\begin{proof}
On the interval $s\in[-R_{\rm LP},R_{\rm LP}]$, direct differentiation of the
scalar map $g$ in \eqref{eq:local-vector-field} gives
\(\sup|g^{(5)}(s)|<140\). Taylor's theorem, \(\|A_G\|_2=6\), and
Hadamard-product bounds give \eqref{eq:nonlinear-majorants}. Let
\(\nu=23/5\) be the slowest transverse decay rate. For a positive integer $m$
and forcing that decays as
\(e^{-m\gamma t}\), the unstable and stable convolutions have the common
bound
\[
 K_m=\left[(\lambda_u+m\gamma)^{-2}
 +(\nu-m\gamma)^{-2}\right]^{1/2}.
\]
Here \(K_1<0.879\), \(K_3<0.699\), and
\(K_1(3.168035R_{\rm LP}^2)<2.785\times10^{-6}\). The contraction theorem therefore
gives \eqref{eq:chart-radius}. For $z\in u^\perp$ in the chart, let
$X_z\in\mathcal X_\gamma$ be the nonlinear Lyapunov--Perron fixed trajectory
with stable initial coordinate $QX_z(0)=z$, let
$X_z^0(t)=e^{(Q\mathcal L_GQ)t}z$ be its linear stable comparison trajectory,
and let $D_z$ denote the Fr\'echet derivative with respect to $z$.
Differentiating the fixed-point equation gives
\(\|X_z-X_z^0\|_\gamma<0.739\|z\|_2^3\) and
\(\|D_zX_z-D_zX_z^0\|_\gamma<2.215\|z\|_2^2\). Subtracting the cubic graph
derivative leaves the fifth-order term and two fixed-point perturbation terms,
so
\[
 C_4\leq
 \frac{35B^5}{\lambda_u+5\gamma}
 +\frac{3.168(B+1)(0.739)B}{\lambda_u+3\gamma}
 +\frac{3.168(2.215)}{\lambda_u+3\gamma}
 <10.387.
\]
This proves \eqref{eq:graph-remainder} without a cutoff outside the stated
ball.
\end{proof}

The complete Laplacian eigenspace projectors
$P_4^G,P_8^G\in\mathbb R^{16\times16}$ are
\begin{equation}
 P_4^G=\frac{A_G+2I_{16}}{4}-\frac J8,
 \qquad
 P_8^G=I_{16}-\frac J{16}-P_4^G.
 \label{eq:projectors}
\end{equation}
Since \(g(s)=(6/5)s-(22/125)s^3+O(s^5)\), the cubic graph is
\begin{equation}
 h_{3,G}(z)=\frac{132}{125}
 \sum_{\lambda_a,\lambda_b,\lambda_c\in\{4,8\}}
 \frac{u^\top\!\left[(P_{\lambda_a}^Gz)\circ(P_{\lambda_b}^Gz)
 \circ(P_{\lambda_c}^Gz)\right]}
 {\lambda_u-\sigma_{\lambda_a}-\sigma_{\lambda_b}-\sigma_{\lambda_c}},
 \label{eq:h3-projector}
\end{equation}
where \(\circ\) denotes componentwise multiplication.  The four denominator
values are
\begin{equation}
 D_{444}=14,
 \qquad D_{448}=\frac{94}{5},
 \qquad D_{488}=\frac{118}{5},
 \qquad D_{888}=\frac{142}{5}.
 \label{eq:denominators}
\end{equation}
Equation \eqref{eq:h3-projector} is invariant under a basis change within a
repeated eigenspace.  The shared eigenvalues fix its denominators, while the
labelled projectors fix the triple overlaps in the numerators.

The symmetry \eqref{eq:signed-symmetry} and uniqueness of the stable graph
imply
\begin{equation}
 h_G(\varepsilon z_\tau)=0
 \quad\text{whenever }\varepsilon z_\tau\text{ lies in the chart}.
 \label{eq:center-on-separatrix}
\end{equation}
Define the scalar local crossing coefficient
$K_G:(0,1/100]\to\mathbb R$ by
\begin{equation}
 K_G(\tau)=-Dh_{3,G}(z_\tau)v_T.
 \label{eq:KG-definition}
\end{equation}
Exact projector contractions give
\begin{align}
 K_\Rook(\tau)
 ={}&\frac{1604394}{172272625}
 +\frac{755568}{172272625}\tau
 -\frac{11806938}{34454525}\tau^2,
 \label{eq:KR}\\
 K_\Shri(\tau)
 ={}&-\frac{1604394}{172272625}
 -\frac{755568}{172272625}\tau
 +\frac{2210274}{6890905}\tau^2.
 \label{eq:KS}
\end{align}
For \(0<\tau\leq1/100\),
\begin{equation}
 K_\Rook(\tau)>\frac{23}{2500},
 \qquad
 K_\Shri(\tau)<-\frac{23}{2500}.
 \label{eq:K-margin}
\end{equation}

A theorem for a finite pulse also requires control of the displacement and the
higher-order remainder. Lemma~\ref{lem:LP-chart} supplies the certified chart
radius and derivative remainder. Let
$\mathcal U_G=\{x\in\mathbb R^{16}:Qx\in\operatorname{dom}(h_G)\}$ and define
$\Psi_G:\mathcal U_G\to\mathbb R$ together with three uniform constants by
\begin{equation}
 \Psi_G(x)=u^\top x-h_G(Qx),
 \qquad
 \overline Z=\frac{1419}{1000},
 \qquad \overline C=2,
 \qquad B_3=\frac{66}{875}.
 \label{eq:Psi-constants}
\end{equation}
Here
$\overline Z\geq\sup_{0<\tau\leq1/100}\|z_\tau\|_2$,
$\overline C=\|v_T\|_2$, and $B_3$ bounds the operator norm of the symmetric
trilinear form associated with the homogeneous cubic $h_{3,G}$, uniformly over
$G\in\{\Rook,\Shri\}$.
For the local pulse sign \(\varsigma\in\{-1,+1\}\), \(0<\tau\leq1/100\),
\(0<\varepsilon\leq1/2000\), and
\(0<\delta\leq\varepsilon/1000\), the complete finite-pulse
remainder obeys
\begin{align}
 &\left|\Psi_G(\varepsilon z_\tau+\varsigma\delta v_T)
 -\varsigma\delta\varepsilon^2K_G(\tau)\right|\notag\\
 &\quad\leq\delta\varepsilon^2\left[
 3B_3\overline Z\overline C^2\frac1{1000}
 +B_3\overline C^3\frac1{10^6}
 +11\overline C\left(\overline Z+\frac{\overline C}{1000}\right)^4
 \frac1{2000^2}\right]\notag\\
 &\quad\leq C_{\rm rem}\delta\varepsilon^2,
 \qquad
 C_{\rm rem}=\frac{18303970734638237}{14000000000000000000}
 <0.001308.
 \label{eq:finite-pulse-remainder}
\end{align}
This bound contains both the nonlinear displacement in \(\delta\) and the
\(h_G-h_{3,G}\) remainder.

\begin{lemma}[Normal-coordinate continuation and orthant entry]
\label{lem:continuation}
Let $G\in\{\Rook,\Shri\}$, let $\varsigma\in\{-1,+1\}$ denote the local pulse
sign, and
take
\[
 x(0)=\varepsilon z_\tau+\varsigma\delta v_T,
 \qquad 0<\tau\leq\frac1{100},\quad
 0<\varepsilon\leq\frac1{2000},\quad
 0<\delta\leq\frac{\varepsilon}{1000}.
\]
For any of the four initial points obtained by pairing each graph with each
pulse sign, write
\[
 x=(h_G(z)+w)u+z,
 \qquad w=\Psi_G(x).
\]
Here $z=Qx\in u^\perp$ is the transverse coordinate and $w\in\mathbb R$ is
the signed normal coordinate. Until it reaches the exit threshold
\(|w|=w_E:=1/4000\), the solution remains in the chart and has an effective
scalar normal growth rate $\Lambda_G(w,z)\in\mathbb R$ satisfying
\begin{equation}
 \dot w=\Lambda_G(w,z)w,
 \qquad
 0.1999968<\Lambda_G(w,z)<0.2000032.
 \label{eq:normal-growth}
\end{equation}
If \(w(0)\neq0\), its sign is preserved and it reaches
\(|w|=w_E\) at a finite time \(T_E>50\). At that time every coordinate has
the sign of \(w\), with componentwise margin greater than
\(6.24996\times10^{-5}\).
\end{lemma}

\begin{proof}
Invariance of the stable graph gives, with
\(x_g=h_G(z)u+z\) and \(\Delta N=N_G(x)-N_G(x_g)\),
\[
 \dot w=\lambda_uw+u^\top\Delta N-Dh_G(z)Q\Delta N.
\]
The cubic graph bound and \eqref{eq:graph-remainder} imply
\[
 \|Dh_G(z)\|_2
 \leq3(66/875)R_{\rm LP}^2+11R_{\rm LP}^4
 <2.263\times10^{-7}.
\]
Combining this with \eqref{eq:nonlinear-majorants} and dividing by nonzero
\(w\) gives \eqref{eq:normal-growth}; continuity supplies the value at
\(w=0\).

Let \(\nu=23/5\). Variation of constants in the transverse equation and the
same nonlinear majorant give
\[
 \|z(t)\|_2
 \leq e^{-\nu t}\|z(0)\|_2+2.296\times10^{-10},
 \qquad
 |h_G(z(t))|<3.31\times10^{-10}.
\]
Together with the initial bound
\(\|z(0)\|_2\leq(1419/1000+2/1000)/2000\), these inequalities yield
\[
 \|x(t)\|_2<0.000960500561<R_{\rm LP},
 \qquad
 \|z(t)\|_2<r_{\rm chart}
\]
as long as \(|w|<w_E\). Thus neither chart boundary can terminate the
bootstrap. Equation \eqref{eq:normal-growth} preserves the sign of \(w\) and
forces finite exit. Moreover,
\(|w(0)|<2.27\times10^{-13}\), while
\(e^{0.2000032\times50}<10^9<w_E/|w(0)|\), so \(T_E>50\).
At exit, \(e^{-\nu T_E}<e^{-230}<10^{-60}\), and therefore
\[
 \frac{w_E}{4}-\|z(T_E)\|_\infty
 -\frac{|h_G(z(T_E))|}{4}
 >6.24996\times10^{-5}.
\]
This proves strict orthant entry with the sign of \(w\).
\end{proof}

\subsection*{Completion of the small-radius proof}
For every
\begin{equation}
 0<\tau\leq\frac1{100},
 \qquad 0<\varepsilon\leq\frac1{2000},
 \qquad 0<\delta\leq\frac{\varepsilon}{1000},
 \label{eq:small-box}
\end{equation}
the scalar flow \eqref{eq:local-vector-field} has the limits
\begin{equation}
\begin{array}{c|cc}
 &\varepsilon z_\tau+\delta v_T&\varepsilon z_\tau-\delta v_T\\
\hline
\Rook&+q_\star\one&-q_\star\one\\
\Shri&-q_\star\one&+q_\star\one
\end{array}
\label{eq:small-winner-table}
\end{equation}
where \(q_\star\in(0,1)\) is the unique positive solution of
\(7q=6g(q)\).

Equations \eqref{eq:K-margin} and \eqref{eq:finite-pulse-remainder} give
\begin{equation}
 |\Psi_G(\varepsilon z_\tau\pm\delta v_T)|
 >0.0078\,\delta\varepsilon^2
 \label{eq:Psi-margin}
\end{equation}
with the sign table \((+,-)\) for the rook graph and \((-,+)\) for the
Shrikhande graph. Lemma~\ref{lem:continuation} preserves each sign and proves
strict entry into the corresponding orthant. The system is cooperative, and
comparison of the scalar minimum and maximum then gives the two consensus limits
\cite{Smith1995}. The exit time is finite for every \(\delta>0\), but it is not
uniform as \(\delta\downarrow0\).

\section{Independent interval certificate at moderate amplitude}
\label{sec:moderate}

The registered operational construction uses
\((\varrho,\tau,\mu)=(4/5,1/100,1/100)\), where the three parameters are the
background amplitude, center deformation and pulse amplitude defined in
Section~S4. Because this point lies far outside the local box
\eqref{eq:small-box}, we certify its winner table independently by interval
arithmetic rather than extrapolate \eqref{eq:h3-projector} or
\eqref{eq:finite-pulse-remainder}.

For $G\in\{\Rook,\Shri\}$ and arm sign $\eta\in\{-1,+1\}$, let
$o_G^{(\eta)}:[0,\infty)\to[-1,1]^{16}$ be the registered opinion trajectory defined
by
\begin{equation}
 \dot o_G^{(\eta)}(t)=F_G(o_G^{(\eta)}(t)),
 \qquad
 o_G^{(\eta)}(0)=\varrho z_\tau+\eta\mu v_T.
 \label{eq:registered-trajectories}
\end{equation}
The notation $o_G^\pm$ refers to $\eta=\pm1$, and $T=3$ is the certified finite
time horizon.

\begin{lemma}[Validated orthant certificate at finite amplitude]
\label{lem:moderate-certificate}
For the registered trajectories in \eqref{eq:registered-trajectories},
\[
 \min_i o_{\Rook,i}^{(+)}(3)>2.142125325208\times10^{-5},
 \qquad
 \max_i o_{\Shri,i}^{(+)}(3)<-2.142758397374\times10^{-5}.
\]
The signed symmetry gives the reversed inequalities for $\eta=-1$. The same
four winner signs hold throughout the nonzero parameter box
\eqref{eq:moderate-box}, with a simultaneous coordinate margin greater than
$3.19\times10^{-6}$.
\end{lemma}

\begin{proof}
First, \(A_G^2=4I_{16}+2J\) implies
\(\|A_G|_{\one^\perp}\|_2=2\).  The global slope bound
\(0\leq g'(s)\leq6/5\) for $s\in[-1,1]$ implies
\[
 \|Qg(x)\|_2
 =\|Q\{g(x)-g((\one^\top x/16)\one)\}\|_2
 \leq\frac65\|Qx\|_2.
\]
Together with the transverse adjacency norm, this gives
\begin{equation}
 \frac{d}{dt}\frac12\|Qx(t)\|_2^2
 \leq-\frac{23}{5}\|Qx(t)\|_2^2.
 \label{eq:transverse-contraction}
\end{equation}
At \(T=3\), this yields the certified transverse bound
\begin{equation}
 \|Qo_G^\pm(3)\|_2<1.2\times10^{-6}.
 \label{eq:transverse-endpoint}
\end{equation}

The enclosure over the finite time interval uses the state envelope radius
$R_{\rm enc}:=7/8$ and the exact algebraic relation
\begin{equation}
 (1+y)^5(1-\xi)^6-(1-y)^5(1+\xi)^6=0,
 \qquad \xi\in[-R_{\rm enc},R_{\rm enc}],\quad y=g(\xi)\in[-1,1],
 \label{eq:algebraic-g}
\end{equation}
and exact integer sign tests to validate every root bracket.  The integration
settings follow the principle of outward enclosure in validated numerics
\cite{Tucker2011} and are
\begin{equation}
 \Delta t=\frac1{1000},
 \qquad n_{\rm step}=3000,
 \label{eq:interval-settings}
\end{equation}
where $R_{\rm enc}$ is the state envelope bound
$\|o_G^{(\eta)}(t)\|_\infty<R_{\rm enc}$ for
$G\in\{\Rook,\Shri\}$ and $\eta\in\{-1,+1\}$, $\Delta t$ is the time step, and
$n_{\rm step}$ is the number of steps, so $n_{\rm step}\Delta t=T=3$.
For the four registered trajectories, one for each graph and arm, the Taylor
enclosure of degree four uses the derivative bound and local Taylor remainder
\begin{equation}
\begin{aligned}
M_5&:=\max_{G,\eta}\sup_{0\leq t\leq3}
\left\|\frac{d^5o_G^{(\eta)}(t)}{dt^5}\right\|_\infty
<1.397\times10^8,\\
R_5^{\rm Tay}&:=\frac{M_5(\Delta t)^5}{5!}<1.164\times10^{-9}.
\end{aligned}
\end{equation}

Let $E_k^{\rm num}\in[0,\infty)$ be the radius of the accumulated numerical
state error after $k$ steps and let $\omega_k^{\rm eval}\in[0,\infty)$ be the
outward radius for function evaluation at step $k$. Conversion of the exact
rational initial state to binary64 gives
$E_0^{\rm num}<4.45\times10^{-17}$. For
$k=0,\ldots,n_{\rm step}-1$, the checked recurrence is
\begin{equation}
 E_{k+1}^{\rm num}\leq\frac{5000}{4999}E_k^{\rm num}
 +R_5^{\rm Tay}+\omega_k^{\rm eval}.
 \label{eq:error-recurrence}
\end{equation}
The checker verifies $\omega_k^{\rm eval}<5.551\times10^{-16}$ at every step
and $E_{n_{\rm step}}^{\rm num}<4.783\times10^{-6}$ for both graphs.
For the positive-pulse arm, the endpoint enclosures are
\begin{equation}
\begin{array}{lccc}
\toprule
G&\text{mean magnitude lower bound}&\text{transverse norm upper bound}
 &\text{coordinate margin}\\
\midrule
\Rook&2.2621\times10^{-5}&1.2\times10^{-6}&2.1421\times10^{-5}\\
\Shri&2.2627\times10^{-5}&1.2\times10^{-6}&2.1427\times10^{-5}\\
\bottomrule
\end{array}
\label{eq:moderate-margins}
\end{equation}
The Shrikhande entries in \eqref{eq:moderate-margins} are magnitudes for its
negative mean and negative orthant.  Odd symmetry supplies the two
negative-pulse arms.  Strict orthant entry followed by cooperative comparison
proves the same table as \eqref{eq:small-winner-table} at the registered
moderate point.  The scalar enclosure also remains valid on the nonzero box
\begin{equation}
 \left[\frac{79999}{100000},\frac{80001}{100000}\right]_\varrho
 \times
 \left[\frac{4999}{500000},\frac1{100}\right]_\tau
 \times
 \left[\frac1{100},\frac{1001}{100000}\right]_\mu.
 \label{eq:moderate-box}
\end{equation}
Every point in this box starts within $10^{-5}$ in the sup norm of its
registered anchor. The flow amplifies this deviation by less than
$e^{(1/5)3}<1.823$ up to $T=3$, leaving a simultaneous strict coordinate
margin greater than $3.19\times10^{-6}$ on both graphs.
\end{proof}

\subsection*{Lift to the main winner theorem}
Apply the atomic crossed schedule in Section~S4, followed by
\eqref{eq:N1} with \(W=A_G\), \(\alpha=\kappa=1\), and \(\beta=6/5\).
For the registered parameters \eqref{eq:registered-parameters}, the node
density matrices converge as follows:
\begin{equation}
\begin{array}{c|cc}
 &\eta=+1&\eta=-1\\
\hline
\Rook&\tfrac12(I_2+q_\star\sigma_z)&\tfrac12(I_2-q_\star\sigma_z)\\[1mm]
\Shri&\tfrac12(I_2-q_\star\sigma_z)&\tfrac12(I_2+q_\star\sigma_z)
\end{array}
\label{eq:density-winner-table}
\end{equation}
at every node, where
\(q_\star\in(0,1)\) is the unique positive root of $7q=6g(q)$ and
$q_\star=0.4338978338\ldots$.
Lemma~\ref{lem:moderate-certificate} proves the scalar opinion table. The exact
semiconjugacy \eqref{eq:semiconjugacy} identifies those trajectories with the
opinion coordinate of the flow on density states. Equation
\eqref{eq:context-decay} gives \(c_i(t)=e^{-7t}c_i(0)\to0\), producing the
limits in \eqref{eq:density-winner-table}.

\paragraph{Macroscopic parity-twirl control.}
Replace both local channels by \eqref{eq:parity-twirl}, while retaining the
same nodewise preparations and crossed schedules.  Both arms then start the
network flow with the identical opinion vector
\((4/5)z_{1/100}\).  This vector lies in
\(\Fix(-P_T):=\{x\in\mathbb R^{16}:-P_Tx=x\}\subset\one^\perp\), and that
subspace is invariant. Inequality
\eqref{eq:transverse-contraction} therefore sends the opinion vector to zero.
The context coordinates also decay, so every node density matrix converges
to \(I_2/2\) in both arms.

\section{Proof of the open-neighborhood corollary for channel parameters}
\label{sec:robustness}

For a unit axis \(n\in S^1\) and a partial-dephasing fraction
\(\chi\in(0,1)\), write
\begin{equation}
 \mathcal E_{n,\chi}=(1-\chi)\Id+\chi\Dens_n.
 \label{eq:channel-family}
\end{equation}
Here $\Dens_n:\mathbb M_2\to\mathbb M_2$ is the complete-dephasing map
defined by \eqref{eq:complete-dephasing} with axis $n$.
Let $\mathring{\mathbb B}_2:=\{r\in\mathbb R^2:\|r\|_2<1\}$ be the open Bloch
disk. Consider the declared local parameter space
\begin{equation}
 \mathcal P=(S^1)^2\times(0,1)^2\times\mathring{\mathbb B}_2,
 \label{eq:parameter-space}
\end{equation}
whose entries are the two axes, the two dephasing fractions, and the common
preparation at the active nodes. The registered point is
$\theta_0=(n_A,n_B,1/2,1/2,r_{\rm pre})\in\mathcal P$.

\paragraph{Claim used in the main text.}
Keep the two labelled graphs, the crossed order schedules, and the nonlinear
flow fixed.  There is a relative open neighborhood
\(U\subset\mathcal P\) of \(\theta_0\) such that every
\(\theta\in U\) preserves all four entries of the winner table for density states
\eqref{eq:density-winner-table}.  After reducing \(U\) if necessary, the two
local channels remain noncommuting and their two ordered compositions retain
a nonzero opinion contrast on the common preparation.

\begin{proof}[Proof of the main-text corollary]
The strict decrease of \(g(q)/q\) at its positive crossing makes the
consensus eigenvalue at \(q_\star\) negative.  The two transverse adjacency
eigenvalues are \(2\) and \(-2\); since \(0\leq g'\leq6/5\), their linear rates
are also negative.  The contextual eigenvalues equal \(-7\).  Hence
the two nodewise tuples
\[
 \boldsymbol\rho_\pm^\star
 =\left(\frac12(I_2\pm q_\star\sigma_z)\right)_{i\in V}
 \in\mathcal D_{\rm R}^{16}
\]
are asymptotically stable equilibria of the complete flow and their
basins are open.  For each graph and arm, finite channel composition
and the embedding in Section~S4 depend continuously on
\(\theta\in\mathcal P\).  At \(\theta_0\), the resulting four initial states
belong to the four required open basins by
the table certified by the interval calculation
\eqref{eq:density-winner-table}. The intersection of their four basin
preimages is a relative open neighborhood of \(\theta_0\).  The commutator and
the ordered opinion contrast are continuous functions of \(\theta\) and are
nonzero at \(\theta_0\), so a further intersection preserves both properties.
\end{proof}

The statement establishes existence only. Quantifying the radius or
incorporating uncertainty in the graph, schedule, or flow after the kick would
require additional analysis.

\section{Reproduction}
\label{sec:reproduction}

The exact algebra and the independent winner certificates can be reproduced
from the arXiv submission source root with
\begin{verbatim}
python3 -B anc/verify_nonconjugate_separatrix_v0_1.py

python3 -B anc/verify_moderate_amplitude_parameter_box_v0_1.py

python3 -B anc/paper1_rebit_theory/verify_rebit_closure_v0_1.py \
  --with-winner
\end{verbatim}
The first command checks the graph identities, projector contractions, the
bound for finite pulses, and continuation from a small radius. The second
checks the moderate scalar parameter box. The unified third command checks the
CPTP channel algebra, both controls, the crossed network embedding, the
parameter lock, the moderate interval winner certificate, and the exact lift
from opinions to density states. These scripts make no LLM calls and require no
empirical calibration.

{\small
\bibliographystyle{unsrt}
\bibliography{references}
}